%% file: paper.tex
\documentclass{llncs}

\input{packages}
\input{tikz-defs}
\input{macros}

\renewcommand{\circ}{\bigcirc}
\newcommand{\myspace}{\vspace*{-1em}}
\newcommand{\myspaceb}{\vspace*{-0.5em}}

\title{Value Iteration for Simple Stochastic Games: Stopping Criterion and Learning Algorithm\thanks{This research was funded in part by 
the Studienstiftung des deutschen Volkes project ``Formal methods for analysis of attack-defence diagrams'', the Czech Science Foundation grant No.~\mbox{18-11193S}, and the German Research Foundation (DFG) project KR 4890/2-1 ``Statistical Unbounded Verification''.}}
\author{Edon Kelmendi \and Julia Kr\"amer \and Jan K\v{r}et{\'i}nsk{\'y} \and Maximilian Weininger}
\institute{Technical University of Munich} 

\begin{document}

\maketitle

\vspace{-2em}

\begin{abstract}
Simple stochastic games can be solved by value iteration (VI), which yields a sequence of under-approximations of the value of the game.
This sequence is guaranteed to converge to the value only in the limit.
Since no stopping criterion is known, this technique does not provide any guarantees on its results.
We provide the first stopping criterion for VI on simple stochastic games.
It is achieved by additionally computing a convergent sequence of \emph{over-approximations} of the value, relying on an analysis of the game graph.
Consequently, VI becomes an anytime algorithm returning the approximation of the value and the current error bound.
As another consequence, we can provide a simulation-based asynchronous VI algorithm, which yields the same guarantees, but without necessarily exploring the whole game graph.
\end{abstract} 

\input{1_intro}

\input{2_prelim}

\input{3_example}

\input{4_main}

\input{5_ImplAndRes}

\input{conclusion}

\newpage

\bibliographystyle{alpha}
\bibliography{ref}

\appendix
\section*{Appendix}
\input{app}

\end{document}

%% file: packages.tex
\usepackage{amsmath,amssymb,xcolor,graphicx,wrapfig,paralist}
\usepackage{algorithmicx,algorithm}
\usepackage[noend]{algpseudocode}
\usepackage{booktabs}
\usepackage{proof}
\usepackage{tikz}
\usepackage{csquotes}
\usepackage{graphics}
\usepackage{stmaryrd}
\usepackage[caption=false]{subfig}

\usepackage{ellipsis, mparhack, ragged2e} 
\usepackage[l2tabu, orthodox]{nag} 

\usepackage{tabu,multirow}
\usepackage{xparse}
\usepackage{mathtools}
\usepackage{environ}

\usepackage[english]{babel}

\usepackage{caption}

\usepackage{algorithm}
\usepackage{algpseudocode}

\usepackage{marginfix}
\usepackage{xifthen}

\usepackage{listings}
\lstdefinestyle{customJ}{
  belowcaptionskip=1\baselineskip,
  breaklines=true,
  frame=L,
  xleftmargin=\parindent,
  language=Java,
  showstringspaces=false,
  basicstyle=\footnotesize\ttfamily,
  keywordstyle=\bfseries\color{green!40!black},
  commentstyle=\itshape\color{purple!40!black},
  identifierstyle=\color{blue},
  stringstyle=\color{orange},
}

\usepackage{afterpage}

\usepackage{xfrac}

%% file: tikz-defs.tex
\usetikzlibrary{shapes,positioning,arrows,calc,automata,matrix,fit}
\tikzstyle{state}+=[minimum size = 6mm, inner sep=0,outer sep=1]
\tikzset{->,>=stealth'}

%% file: macros.tex
\newcommand{\todojan}[1]{\todo{\textbf{Jan:\\}#1}}

\algrenewcommand{\algorithmiccomment}[1]{\hskip1.5em \textbackslash *  #1 * \textbackslash}

\newcolumntype{L}{X}
\newcolumntype{R}{>{\raggedleft\arraybackslash}X}
\newcolumntype{C}{>{\centering\arraybackslash}X}

\DeclareGraphicsExtensions{.pdf, .png, .jpg}

\NewDocumentCommand{\todo}{m}{%
	\begin{tikzpicture}[remember picture, baseline=-0.75ex]%
	\node [coordinate] (inText) {};%
	\end{tikzpicture}%
	\marginpar{%
		\begin{tikzpicture}[remember picture, font=\scriptsize]%
		\draw node[draw=red, text width = 3.5cm, inner sep=0.3mm] (inNote){#1};%
		\end{tikzpicture}%
	}%
	\begin{tikzpicture}[remember picture, overlay]%
	\draw[draw=red]%
	([yshift=-0.2cm] inText)%
	-| (inNote.south);%
	\end{tikzpicture}%
}%

\RequirePackage{marvosym}

\DeclarePairedDelimiter{\delimabs}{\lvert}{\rvert}
\DeclarePairedDelimiter{\delimnorm}{\lVert}{\rVert}
\DeclarePairedDelimiter{\delimpospart}{\lgroup}{\rgroup^+}

\DeclarePairedDelimiterX{\deliminner}[2]{\lange}{\rangle}{#1, #2}
\DeclarePairedDelimiter{\delimcardinality}{\lvert}{\rvert}
\DeclarePairedDelimiter{\delimset}{\lbrace}{\rbrace}
\DeclarePairedDelimiter{\delimtuple}{(}{)}
\DeclarePairedDelimiter{\delimlistt}{[}{]}
\DeclarePairedDelimiter{\delimfun}{(}{)}

\NewDocumentCommand{\abs}{sm}{\IfBooleanTF{#1}{\delimabs{#2}}{\delimabs*{#2}}}
\NewDocumentCommand{\norm}{sm}{\IfBooleanTF{#1}{\delimnorm{#2}}{\delimnorm*{#2}}}
\NewDocumentCommand{\pospart}{sm}{\IfBooleanTF{#1}{\delimpospart{#2}}{\delimpospart*{#2}}}
\NewDocumentCommand{\negpart}{sm}{\IfBooleanTF{#1}{\delimnetpart{#2}}{\delimnetpart*{#2}}}
\NewDocumentCommand{\inner}{sm}{\IfBooleanTF{#1}{\deliminner{#2}}{\deliminner*{#2}}}
\NewDocumentCommand{\cardinality}{sm}{\IfBooleanTF{#1}{\delimcardinality{#2}}{\delimcardinality*{#2}}}
\NewDocumentCommand{\set}{sm}{\IfBooleanTF{#1}{\delimset*{#2}}{\delimset{#2}}}
\NewDocumentCommand{\tuple}{sm}{\IfBooleanTF{#1}{\delimtuple{#2}}{\delimtuple*{#2}}}
\NewDocumentCommand{\closure}{sm}{\IfBooleanTF{#1}{\delimclosure{#2}}{\delimclosure*{#2}}}
\NewDocumentCommand{\listt}{sm}{\IfBooleanTF{#1}{\delimlistt{#2}}{\delimlistt*{#2}}}
\NewDocumentCommand{\fun}{smm}{\IfBooleanTF{#1}{{#2}\delimfun{#3}}{{#2}\delimfun*{#3}}}
\NewDocumentCommand{\funMacro}{smm}{\IfNoValueTF{#3}{#1}{\fun{#2}{#3}}}

\DeclareMathOperator{\ExistsOp}{\exists}
\DeclareMathOperator{\ForallOp}{\forall}

\NewDocumentCommand{\Exists}{gg}{\IfNoValueTF{#1}{\ExistsOp}{\ExistsOp #1. \, #2}}
\NewDocumentCommand{\Forall}{gg}{\IfNoValueTF{#1}{\ForallOp}{\ForallOp #1. \, #2}}

\newcommand{\unionSym}{\cup}

\newcommand{\strictunionBin}{\mathbin{\dot{\unionSym}}}
\newcommand{\intersectionSym}{\cap}
\newcommand{\intersectionBin}{\mathbin{\intersectionSym}}
\newcommand{\UnionSym}{\bigcup}

\newcommand{\strictunion}{\strictunionBin}
\newcommand{\intersection}{\intersectionBin}
\newcommand{\Union}{\UnionSym}

\newcommand{\Naturals}{\mathbb{N}}

\newcommand{\Distributions}{\mathcal{D}}

\NewDocumentCommand{\convto}{G{}}{\xrightarrow{#1}}
\NewDocumentCommand{\weakto}{G{}}{\xrightharpoonup{#1}}
\NewDocumentCommand{\weakstarto}{G{}}{\xrightharpoonup[*]{#1}}

 \DeclareDocumentCommand{\diff}{D<>{} O{}  D(){}}{\Delta_{#1}^{#2}\ifthenelse{\isempty{#3}}{}{(#3)}}

\DeclareMathOperator*{\argmax}{arg\, max}
\DeclareMathOperator*{\argmin}{arg\, min}
\DeclareDocumentCommand{\post}{D<>{} O{} D(){}}{\mathsf{Post}_{#1}^{#2}\ifthenelse{\isempty{#3}}{}{(#3)}}
\DeclareMathOperator{\leaves}{\mathbin{\mathop{exits}}}
\newcommand{\leaving}{exiting}
\DeclareMathOperator{\stays}{\mathbin{\mathop{\mathsf{stays\_in}}}}
\newcommand{\eqdef}{\vcentcolon=}

\newcommand{\qee}{\hfill$\triangle$} 

\NewDocumentCommand{\distributions}{d()}{\funMacro{\mathcal{D}}{#1}}

\newcommand{\MECs}{\mathsf{MEC}}

\DeclareDocumentCommand{\val}{D<>{} O{}  D(){} t'}{\mathsf{V}_{#1}^{\IfBooleanTF{#4}{\prime #2}{#2}}\ifthenelse{\isempty{#3}}{}{(#3)}}
\DeclareDocumentCommand{\ub}{D<>{} O{}  D(){} t'}{\mathsf{U}_{#1}^{\IfBooleanTF{#4}{\prime #2}{#2}}\ifthenelse{\isempty{#3}}{}{(#3)}}
\DeclareDocumentCommand{\gub}{D<>{} O{}  D(){} t'}{\mathsf{G}_{#1}^{\IfBooleanTF{#4}{\prime #2}{#2}}\ifthenelse{\isempty{#3}}{}{(#3)}}
\DeclareDocumentCommand{\lb}{D<>{} O{}  D(){} t'}{\mathsf{L}_{#1}^{\IfBooleanTF{#4}{\prime #2}{#2}}\ifthenelse{\isempty{#3}}{}{(#3)}}
\DeclareDocumentCommand{\game}{D<>{} O{} D(){} t'}{\mathsf{G}_{#1}^{\IfBooleanTF{#4}{\prime}{}#2}\ifthenelse{\isempty{#3}}{}{(#3)}}
\DeclareDocumentCommand{\transition}{D<>{} O{} D(){}}{\rightarrow_{#1}^{#2}\ifthenelse{\isempty{#3}}{}{(#3)}}

\newcommand{\Mc}{\mathsf{M}}

\newcommand{\SG}{\textrm{SG}}
\newcommand{\SGs}{\textrm{SGs}}
\DeclareDocumentCommand{\G}{D<>{} O{} t' D(){}}{\mathsf{G}_{#1}^{\IfBooleanTF{#3}{\prime}{}#2}\ifthenelse{\isempty{#4}}{}{(#4)}}
\DeclareDocumentCommand{\exGame}{D<>{} O{} t'  D(){}}{\mathsf{G}_{#1}^{\IfBooleanTF{#3}{\prime}{}#2}\ifthenelse{\isempty{#4}}{}{(#4)}=(\states<#1>[\IfBooleanTF{#3}{\prime}{}#2],\states<\Box\ifthenelse{\isempty{#1}}{}{,#1}>[\IfBooleanTF{#3}{\prime}{}#2],\states<\circ\ifthenelse{\isempty{#1}}{}{,#1}>[\IfBooleanTF{#3}{\prime}{}#2],\istate<#1>[\IfBooleanTF{#3}{\prime}{}#2],\actions<#1>[\IfBooleanTF{#3}{\prime}{}#2],\Av<#1>[\IfBooleanTF{#3}{\prime}{}#2],\trans<#1>[\IfBooleanTF{#3}{\prime}{}#2],\target,\sink)}

\DeclareDocumentCommand{\states}{D<>{} O{}  t'}{\mathit{S}_{#1}^{\IfBooleanTF{#3}{\prime~#2}{#2}}}
\DeclareDocumentCommand{\state}{D<>{} O{}  t'}{\mathsf{s}_{#1}^{\IfBooleanTF{#3}{\prime #2}{#2}}}
\DeclareDocumentCommand{\istate}{D<>{} O{} t'}{\mathsf{s}_{0\ifthenelse{\isempty{#1}}{}{,#1}}^{\IfBooleanTF{#3}{\prime~#2}{#2}}}

\newcommand{\initstate}{\state<0>}

\DeclareDocumentCommand{\trans}{D<>{} O{} t' D(){} D(){}}{\delta_{#1}^{\IfBooleanTF{#3}{\prime}{}#2}\ifthenelse{\isempty{#4}}{}{(#4)}\ifthenelse{\isempty{#5}}{}{(#5)}}
\DeclareDocumentCommand{\Av}{D<>{} O{} t' D(){}}{\mathsf{Av}_{#1}^{\IfBooleanTF{#3}{\prime}{}#2}\ifthenelse{\isempty{#4}}{}{(#4)}}

\DeclareDocumentCommand{\F}{D<>{} O{} t' D(){}}{\mathsf{F}_{#1}^{\IfBooleanTF{#3}{\prime}{}#2}\ifthenelse{\isempty{#4}}{}{(#4)}}

\newcommand{\Path}{\rho}
\DeclareDocumentCommand{\Path}{D<>{} O{} t' D(){}}{\path<#1>[#2]\IfBooleanTF{#3}{'}{}(#4)}
\DeclareDocumentCommand{\path}{D<>{} O{} t' D(){}}{\rho_{#1}^{\IfBooleanTF{#3}{\prime}{}#2}\ifthenelse{\isempty{#4}}{}{(#4)}}
\newcommand{\fpath}{\mathsf{w}}
\DeclareDocumentCommand{\Paths}{D<>{} O{} t' D(){}}{\Omega_{#1}^{\IfBooleanTF{#3}{\prime}{}#2}\ifthenelse{\isempty{#4}}{}{(#4)}}
\newcommand{\straa}{\sigma}

\newcommand{\strab}{\tau}

\DeclareDocumentCommand{\strategy}{D<>{} O{} D(){}
  t*}{{\IfBooleanTF{#4}{\tau}{\sigma}}_{#1}^{#2}\ifthenelse{\isempty{#3}}{}{(#3)}}

\DeclareDocumentCommand{\actions}{D<>{} O{} t' d()}{{\IfNoValueTF{#4}{\mathsf{A}}{\fun{\mathsf{A}}{#4}}}_{#1}^{\IfBooleanTF{#3}{\prime~#2}{#2}}}
\DeclareDocumentCommand{\action}{D<>{} O{} t'}{\mathsf{a}_{#1}^{\IfBooleanTF{#3}{\prime#2}{#2}}}

\newcommand{\mec}{\mathsf{MEC}}
\newcommand{\scc}{\mathsf{SCC}}

\newcommand{\bscc}{\mathsf{BSCC}}

\newcommand{\pr}{\mathbb P}

\newcommand{\best}{\mathsf{best}}

\newcommand{\Vis}{\mathsf{Vis}}
\DeclareDocumentCommand{\target}{D<>{} O{}
  t'}{\mathfrak{1}_{#1}^{\IfBooleanTF{#3}{\prime}{}#2}}
\DeclareDocumentCommand{\fail}{D<>{} O{} t'}{\mathsf{\bot}_{#1}^{\IfBooleanTF{#3}{\prime}{}#2}}

\newcommand{\EXPLOREBREAK}{\mathsf{SIM\_TOO\_LONG}}
\newcommand{\UPDATE}{\mathsf{UPDATE}}
\newcommand{\ADJUST}{\mathsf{DEFLATE}}
\newcommand{\GETSUCC}{\mathsf{GETSUCC}}
\DeclareDocumentCommand{\CEC}{D<>{} O{} D(){}
  t*}{\IfBooleanTF{#4}{\simple~}{}\ifthenelse{\isempty{#3}}{}{#3\textit{-}}\textrm{BEC}_{#1}^{#2}}

\newcommand{\BCEC}{\Box\textit{-}\CEC}
\newcommand{\COLLAPSEFUN}{\mathsf{COLLAPSE}}

\newcommand{\simple}{simple}
\newcommand{\Simple}{Simple}

\NewDocumentCommand{\exit}{D<>{}D[]{}D(){}}{\mathsf{bestExit}_{#1}^{#2}\ifthenelse{\isempty{#3}}{}{(#3)}}

\newcommand{\FIND}{{\textsf{FIND\_MSEC}}}

\newcommand{\rSTA}{\mathsf{runSampleTrial}}

\newcommand{\drawcirc}{\node[draw,circle,minimum size=.7cm, outer sep=1pt]}
\newcommand{\drawbox}{\node[draw,rectangle,minimum size=.7cm, outer sep=1pt]}
\newcommand{\drawdummy}{\node[minimum size=0,inner sep=0]}

\newcommand{\para}[1]{\noindent\textbf{#1} \ }

\newcommand{\sink}{\mathfrak 0}


%% file: 1_intro.tex
\vspace{-2.5em}

\section{Introduction} \label{sec:intro}
\vspace{-0.8em}
\para{Simple stochastic game} (SG) \cite{condonComplexity} is a zero-sum two-player game played on a graph by Maximizer and Minimizer, who choose actions in their respective vertices (also called states).
Each action is associated with a probability distribution determining the next state to move to.
The objective of Maximizer is to maximize the probability of reaching a given target state; the objective of Minimizer is the opposite.

Stochastic games constitute a fundamental problem for several reasons. 
From the theoretical point of view, the complexity of this problem\footnote{Formally, the problem is to decide, for a given $p\in[0,1]$ whether Maximizer has a strategy ensuring probability at least $p$ to reach the target.} is known to be in $\mathbf{UP}\cap\mathbf{coUP}$~\cite{doi:10.1287/mnsc.12.5.359} , but no polynomial-time algorithm is known.
Further, several other important problems can be reduced to SG, for instance parity games, mean-payoff games, discounted-payoff games and their stochastic extensions can all be reduced to SG \cite{DBLP:journals/corr/abs-1106-1232}.
The task of solving SG is also polynomial-time equivalent to solving perfect information Shapley, Everett and Gillette games \cite{DBLP:conf/isaac/AnderssonM09}.
Besides, the problem is practically relevant in verification and synthesis.
SG can model reactive systems, with players corresponding to the controller of the system and to its environment, where quantified uncertainty is explicitly modelled.
This is useful in many application domains, ranging from smart energy management \cite{DBLP:journals/fmsd/ChenFKPS13} to autonomous urban driving \cite{DBLP:conf/qest/ChenKSW13}, robot motion planning \cite{DBLP:journals/algorithmica/LaValle00} to self-adaptive systems \cite{DBLP:conf/icse/CamaraMG14}; for various recent case studies, see e.g. \cite{DBLP:journals/ejcon/SvorenovaK16}.
Finally, since Markov decision processes (MDP) \cite{puterman} are a special case with only one player, SG can serve as abstractions of large MDP \cite{DBLP:journals/fmsd/KattenbeltKNP10}.

\para{Solution techniques}
There are several classes of algorithms for solving SG, most importantly strategy iteration (SI) algorithms \cite{doi:10.1287/mnsc.12.5.359} and value iteration (VI) algorithms \cite{condonComplexity}.
Since the repetitive evaluation of strategies in SI is often slow in practice, VI is usually preferred, similarly to the special case of MDPs \cite{atva17}.
For instance, the most used probabilistic model checker PRISM \cite{prism} and its branch PRISM-Games \cite{PRISM-games} use VI for MDP and SG as the default option, respectively.
However, while SI is in principle a precise method, VI is an approximative method, which converges only in the limit.
Unfortunately, there is no known stopping criterion for VI applied to SG.
Consequently, there are no guarantees on the results returned in finite time.
Therefore, current tools stop when the difference between the two most recent approximations is low, and thus may return arbitrarily imprecise results \cite{BVI}.

\para{Value iteration with guarantees} 
In the special case of MDP,
in order to obtain bounds on the imprecision of the result, one can employ a \emph{bounded} variant of VI \cite{BRTDP,atva} (also called \emph{interval iteration} \cite{BVI}).
Here one computes not only an under-approximation, but also an over-approximation of the actual value as follows.
On the one hand, iterative computation of the least fixpoint of Bellman equations yields an under-approximating sequence converging to the value.
On the other hand, iterative computation of the greatest fixpoint yields an over-approximation, which, however, does not converge to the value.
Moreover, it often results in the trivial bound of $1$.
A solution suggested for MDPs \cite{atva,BVI} is to modify the underlying graph, namely to collapse end components.
In the resulting MDP there is only one fixpoint, thus the least and greatest fixpoint coincide and both approximating sequences converge to the actual value.
In contrast, for general SG no procedure where the greatest fixpoint converges to the value is known.
In this paper we provide one, yielding a stopping criterion.
We show that the pre-processing approach of collapsing is not applicable in general and provide a solution on the original graph.
We also characterize SG where the fixpoints coincide and no processing is needed. 
The main technical challenge is that states in an end component in SG can have different values, in contrast to the case of MDP.

\para{Practical efficiency using guarantees}
We further utilize the obtained guarantees to practically improve our algorithm.
Similar to the MDP case~\cite{atva}, the quantification of the error allows for ignoring parts of the state space, and thus a speed up without jeopardizing the correctness of the result.
Indeed, we provide a technique where some states are not explored and processed at all, but their potential effect is still taken into account in the lower and upper bound on the result.
The information is further used to decide the states to be explored next and the states to be analyzed in more detail.
To this end, simulations and learning are used as tools.
While for MDP this idea has already demonstrated speed ups in orders of magnitude \cite{atva,cav17},
this paper provides the first technique of this kind for \SG.

\para{Our contribution}
 is summarized as follows: \vspace*{-0.5em}
\begin{itemize}
	\item 
We introduce a VI algorithm yielding both under- and over-approximation sequences, both of which converge to the value of the game.
Thus we present the first stopping criterion for VI on SG and the first anytime algorithm with guaranteed precision.
We also characterize when a simpler solution is sufficient.
	\item
We provide a learning-based algorithm, which preserves the guarantees, but is in some cases more efficient since it avoids exploring the whole state space.
	\item
We evaluate the running times of the algorithms experimentally, concluding that obtaining guarantees requires an overhead that is either negligible or mitigated by the learning-based approach. 
\end{itemize}

\vspace{-0.5em}
\para{Related work}
The works closest to ours are the following. As mentioned above, \cite{atva,BVI} describe the solution to the special case of MDP. While \cite{atva} also provides a learning-based algorithm, \cite{BVI} discusses the convergence rate and the exact solution. 
The basic algorithm of \cite{BVI} is implemented in PRISM \cite{DBLP:conf/cav/Baier0L0W17} and the learning approach of \cite{atva} in \textsc{Storm} \cite{DBLP:conf/cav/DehnertJK017}.
The extension for SG where the interleaving of players is severely limited (every end component belongs to one player only) is discussed in \cite{MatPhD}.

Further, in the area of probabilistic planning, bounded real-time dynamic programming \cite{BRTDP} is related to our learning-based approach. 
However, it is limited to the setting of stopping MDP where the target sink or the non-target sink is reached almost surely under any pair of strategies and thus the fixpoints coincide.
Our algorithm works for general SG, not only for stopping ones, without any blowup.

For SG, the tools implementing the standard SI and/or VI algorithms are PRISM-games \cite{PRISM-games}, 
GAVS+ \cite{DBLP:conf/tacas/ChengKLB11} and GIST \cite{DBLP:conf/cav/ChatterjeeHJR10}.
The latter two are, however, neither maintained nor accessible via the links provided in their publications any more.

Apart from fundamental algorithms to solve SG, there are various practically efficient heuristics that, however, provide none or weak guarantees, often based on some form of learning \cite{DBLP:journals/ai/BrafmanT00,DBLP:conf/ijcnn/LiL08,DBLP:conf/ijcai/WenT16,DBLP:conf/saso/TcheukamT16,DBLP:journals/tac/ArslanY17,DBLP:journals/tsmc/BusoniuBS08}.
Finally, the only currently available way to obtain any guarantees through VI is to perform 
$\gamma^2$ iterations and then round to the nearest multiple of $1/\gamma$, yielding the value of the game with precision $1/\gamma$ \cite{visurvey};
here $\gamma$ cannot be freely chosen, but it is a fixed number, exponential in the number of states and the used probability denominators.
However, since the precision cannot be chosen and the number of iterations is always exponential, this approach is infeasible even for small games.


\para{Organization of the paper} Section \ref{sec:prelim} introduces the basic notions and revises value iteration. 
Section~\ref{sec:example} explains the idea of our approach on an example. 
Section~\ref{sec:main} provides a full technical treatment of the method as well as the learning-based variation.
Section~\ref{sec:exper} discusses experimental results and Section~\ref{sec:concl} concludes.
The appendix gives technical details on the pseudocode as well as the conducted experiments and provides more extensive proofs to the theorems and lemmata; in the main body, there are only proof sketches and ideas.
\vspace{0.7em}


%% file: 2_prelim.tex
\section{Preliminaries} \label{sec:prelim}
\vspace*{-0.8em}
\subsection{Basic definitions}

A 
{probability distribution} on a finite set $X$ is a mapping $\trans: X \to [0,1]$, such that $\sum_{x\in X} \trans(x) = 1$.
The set of all probability distributions on $X$ is denoted by $\Distributions(X)$.
%
%
Now we define stochastic games, in literature often referred as simple stochastic games or stochastic two-player games with a reachability objective.

\begin{definition}[\SG]
	A \emph{stochastic game ($\SG$)} is a tuple 
	$(\states,\states<\Box>,\states<\circ>,\initstate,\actions,\Av,\delta,\mathfrak 1,\mathfrak 0)$,
	where $\states$ is a finite set of \emph{states} partitioned
	\ into the sets $\states<\Box>$ and $\states<\circ>$ of states of the player \emph{Maximizer} and \emph{Minimizer}, respectively, 
	$\initstate,\target,\sink \in \states$ is the \emph{initial} state, \emph{target} state, and \emph{sink} state, respectively, $\actions$ is a finite set of \emph{actions}, $\Av: \states \to 2^{\actions}$ assigns to every state a set of \emph{available} actions, and $\trans: \states \times \actions \to \distributions(\states)$ is a \emph{transition function} that given a state $\state$ and an action $\action\in \Av(\state)$ yields a probability distribution over \emph{successor} states.
	
	A \emph{Markov decision process (MDP)} is a special case of $\SG$ where $\states<\circ> = \emptyset$. 
\end{definition}
We assume that $\SGs$ are non-blocking, so for all states $\state$ we have $\Av(\state) \neq \emptyset$.
Further, $\target$ and $\sink$ only have one action and it is a
self-loop with probability~$1$.
Additionally, we can assume that the SG is preprocessed so that all states with no path to $\target$ are merged with $\sink$.

For a state $\state$ and an available action $\action \in \Av(\state)$, we denote the set of successors by $\post(\state,\action) := \set{\state' \mid \trans(\state,\action,\state') > 0}$.
Finally, for any set of states $T \subseteq \states$, we use $T_\Box$ and $T_\circ$ to denote the states in $T$ that belong to Maximizer and Minimizer, whose states are drawn in the figures as $\Box$ and $\circ$, respectively. 
An example of an $\SG$ is given in Figure~\ref{SGex}. 
 

\input{SGex}


The semantics of SG is given in the usual way by means of strategies and the induced Markov chain and the respective probability space, as follows.
An \emph{infinite path} $\path$ is an infinite sequence $\path = \state<0> \action<0> \state<1> \action<1> \dots \in (\states \times \actions)^\omega$, such that for every $i \in \Naturals$, $\action<i>\in \Av(\state<i>)$ and $\state<i+1> \in \post(\state<i>,\action<i>)$.
\emph{Finite path}s are defined analogously as elements of $(\states \times \actions)^\ast \times \states$.
%
Since this paper deals with the reachability objective, we can restrict our attention to memoryless strategies, which are optimal for this objective.
We still allow randomizing strategies, because they are needed for the learning-based algorithm later on.
A \emph{strategy} of Maximizer or Minimizer is a function $\straa: \states<\Box> \to \distributions(\actions)$ or $\states<\circ> \to \distributions(\actions)$, respectively, such that $\straa(\state) \in \distributions(\Av(\state))$ for all $\state$.
We call a strategy \emph{deterministic} if it maps to Dirac distributions only.
Note that there are finitely many deterministic strategies.
A pair $(\straa,\strab)$ of strategies of Maximizer and Minimizer induces a Markov chain $\G[\straa,\strab]$ where the transition probabilities are defined as $\trans(\state, \state') = \sum_{\action \in \Av(\state)} \straa(\state, \action) \cdot \trans(\state, \action, \state')$ for states of Maximizer and analogously for states of Minimizer, with $\straa$ replaced by $\strab$.
The Markov chain 
induces a unique probability distribution $\pr_{\state}^{\straa,\strab}$ over measurable sets of infinite paths \cite[Ch.~10]{BaierBook}. 

We write $\Diamond \target:=\set{\path \mid \exists i \in \Naturals.~\path(i)=\target}$ to denote the (measurable) set of all paths which eventually reach $\target$. 
For each $\state\in\states$, we define
the \emph{value} in $\state$ as 
\[\val(\state) \eqdef \sup_{\straa} \inf_{\strab} \pr_{s}^{\straa,\strab}(\Diamond \target)= \inf_{\strab} \sup_{\straa}\pr_{s}^{\straa,\strab}(\Diamond \target),\]\myspace

\noindent where the equality follows from~\cite{leastReadablePaperJanEverRead}.
We are interested not only in $\val(\initstate)$, but also its $\varepsilon$-approximations and the corresponding ($\varepsilon$-)optimal strategies for both players.


\medskip

Now we recall a fundamental tool for analysis of MDP called end components.
We introduce the following notation.
	Given a set of states $T \subseteq \states$, a state $\state \in T$ and an action $\action \in \Av(\state)$, we say that $(\state,\action) \leaves T$ if $\post(\state,\action) \not\subseteq T$.
We define an end component of a SG as the end component of the underlying MDP with both players unified.
\begin{definition}[EC]
\label{def:EC}
A non-empty set $T\subseteq \states$ of states is an \emph{end component (EC)} if there is a non-empty set $B \subseteq \Union_{\state \in T} \Av(s)$ of actions such that 
	\begin{enumerate}
		\item for each $\state \in T, \action \in B \intersection \Av(\state)$ we do \emph{not} have $(\state,\action) \leaves T$,
		\item for each $\state, \state' \in T$ there is a finite path $\fpath = \state \action<0> \dots \action<n> \state' \in (T \times B)^* \times T$, i.e. the path stays inside $T$ and only uses actions in $B$.
	\end{enumerate}
\end{definition}
Intuitively, ECs correspond to bottom strongly connected components of the Markov chains induced by possible strategies, so for some pair of strategies all possible paths starting in the EC remain there. 
An end component $T$ is a \emph{maximal end component (MEC)} if there is no other end component $T'$ such that $T \subseteq T'$.
Given an $\SG$ $\G$, the set of its MECs is denoted by $\mec(\G)$ and can be computed in polynomial time~\cite{CY95}.

\subsection{(Bounded) value iteration}

The value function $\val$ satisfies the following system of equations, which is referred to as the \emph{Bellman equations}:
\begin{equation}\label{eq:Vs}
\val(\state) =  \begin{cases} \max_{\action \in \Av(\state)}\val(\state,\action)		&\mbox{if } \state \in \states<\Box>  \\
\min_{\action \in \Av(\state)}\val(\state,\action) &\mbox{if } \state \in \states<\circ>\\
1 &\mbox{if } \state =\target \\
0 &\mbox{if } \state =\sink 
\end{cases}
\end{equation}
where\footnote{\label{fn:overload}Throughout the paper, for any function $f:S\to[0,1]$ we overload the notation and also write $f(\state,\action)$ meaning $\sum_{s' \in S} \trans(\state,\action,\state') \cdot f(\state')$.}
\begin{eqnarray}\label{eq:Vsa}
\val(\state,\action) \eqdef \sum_{s' \in S} \trans(\state,\action,\state') \cdot \val(\state')
\end{eqnarray}
Moreover, $\val$ is the \emph{least} solution to the Bellman equations, see e.g.~\cite{visurvey}. 
To compute the value of $\val$ for all states in an $\SG$, one can thus utilize the iterative approximation method
%
%
\emph{value iteration (VI)} as follows.
We start with a lower bound function $\lb<0>\colon\states \to [0, 1]$ such that $\lb<0>(\target)=1$ and, for all other $\state \in \states$, $\lb<0>(\state) 
=0$.
Then we repetitively apply Bellman updates (\ref{eq:Lsa}) and (\ref{eq:Ls}) 
\begin{align}
\lb<n>(\state,\action) &\eqdef \sum_{s' \in S} \trans(\state,\action,\state') \cdot \lb<n-1>(\state')\label{eq:Lsa}\\
\lb<n>(\state) &\eqdef  
\begin{cases} 
\max_{\action \in \Av(\state)}\lb<n>(\state,\action) &
\mbox{if } \state \in \states<\Box>  \\
\min_{\action \in \Av(\state)}\lb<n>(\state,\action) &
\mbox{if } \state \in \states<\circ>
\end{cases}\label{eq:Ls}
\end{align}
until convergence.
%
Note that convergence may happen only in the limit even for such a simple game as in Figure \ref{SGex}.
The sequence is monotonic, at all times a \emph{lower} bound on $V$, i.e. $\lb<i>(\state)\leq\val(\state)$ for all $\state\in\states$, and the least fixpoint satisfies $\lb[*]:=\lim_{n \to \infty} \lb<n> = \val$.


Unfortunately, there is no known stopping criterion, i.e.\ no guarantees how close the current under-approximation is to the value \cite{BVI}.
The current tools stop when the difference between two successive approximations is smaller than a certain threshold, which can lead to arbitrarily wrong results \cite{BVI}.

For the special case of MDP, it has been suggested to also compute the greatest fixpoint \cite{BRTDP} and thus an \emph{upper} bound as follows. The function $\gub: \states \to [0, 1]$ is initialized for all states $\state \in \states$ as $\gub<0>(\state) = 1$ except for 
$\gub<0>(\sink) = 0$. 
Then we repetitively apply updates (\ref{eq:Lsa}) and (\ref{eq:Ls}), where $\lb$ is replaced by $\gub$.
The resulting sequence $\gub<n>$ is monotonic, provides an upper bound on $\val$ and the greatest fixpoint $\gub[*]:=\lim_n \gub<n>$ is the greatest solution to the Bellman equations on $[0,1]^S$.

This approach is called \emph{bounded value iteration (BVI)} (or \emph{bounded real-time dynamic programming (BRTDP)} \cite{BRTDP,atva} or \emph{interval iteration} \cite{BVI}).
If $\lb[*]=\gub[*]$ then they are both equal to $\val$ and we say that \emph{BVI converges}.
BVI is guaranteed to converge in MDP if the only ECs are those of $\target$ and $\sink$ \cite{atva}.
Otherwise, if there are non-trivial ECs they have to be ``collapsed''\footnote{All states of an EC are merged into one, all leaving actions are preserved and all other actions are discarded. 
For more detail see Appendix \ref{app:coll}}.
Computing the greatest fixpoint on the modified MDP results in another sequence $\ub<i>$ of upper bounds on $\val$, converging to $\ub[*]:=\lim_n \ub<n>$.
Then BVI converges even for general MDPs, $\ub[*]=\val$ \cite{atva}, when transformed this way.
The next section illustrates this difficulty and the solution through collapsing on an example.

In summary, all versions of BVI discussed so far and later on in the paper follow the pattern of Algorithm~\ref{alg:gen}.
In the naive version, $\UPDATE$ just performs the Bellman update on $\lb$ and $\ub$  according to Equations (\ref{eq:Lsa})~and~(\ref{eq:Ls}).\footnote{For the straightforward pseudocode, see Algorithm~\ref{alg:UPDATEVI} in Appendix~\ref{sec:UPDATEVI}.}
For a general MDP, 
$\ub$ does not converge to $\val$, but to $\gub[*]$, and thus the termination criterion may never be met if $\gub[*](\initstate)-\val(\initstate) > 0$.
If the ECs are collapsed in pre-processing then $\ub$ converges to $\val$.

For the general case of \SG, the collapsing approach fails and this paper provides another version of BVI where $\ub$ converges to $\val$, based on a more detailed structural analysis of the game.

	\myspace\myspaceb

\begin{algorithm}[H]

	\caption{Bounded value iteration algorithm}\label{alg:gen}
	\begin{algorithmic}[1]
		\Procedure{BVI}{precision $\epsilon>0$}
		\For {$\state \in \states$} \Comment{Initialization}
		\State $\lb(\state)=0$ ~~\Comment{Lower bound}
		\State $\ub(\state)=1$ ~~\Comment{Upper bound}
		\EndFor
		\State $\lb(\target) = 1$ ~~~~~~~\Comment{Value of sinks is determined a priori}
		\State $\ub(\mathfrak 0)=0$\medskip
		
		\Repeat
		\State $\UPDATE(\lb,\ub)$ ~~~~~~\Comment{Bellman updates or their modification}
		\Until{$\ub(\initstate) - \lb(\initstate) < \epsilon$}
		\Comment{Guaranteed error bound}
		\EndProcedure
	\end{algorithmic}

\end{algorithm}

\myspace\myspace



%% file: SGex.tex
\begin{figure}[t]

\myspace\myspace
\centering
\begin{tikzpicture}
\drawdummy (init) at (0,0) {};
\drawcirc (p) at (1,0) {$\mathsf{p}$};
\drawbox (q) at (3,0) {$\mathsf{q}$};
\drawdummy (mid) at (4,0) {};
\drawbox (1) at (6,0.5) {$\target$ };
\drawcirc (0) at (6,-0.5)  {$\sink$};

\draw[->] (init) to (p);
\draw[->]  (p) to[bend left] node [midway,anchor=south] {$\mathsf{a}$}(q) ;
\draw[->]  (q) to [bend left] node [midway,anchor=north] {$\mathsf{b}$} (p);
\draw[-] (q) to node [midway,anchor=south] {$\mathsf{c}$} (mid) ;
\draw[->] (mid) to node [below] {$\sfrac13$} (0);
\draw[->] (mid) to node [above] {$\sfrac13$} (1);
\draw[->] (mid) to [bend left,out=270,in=270] node[above] {$\sfrac13$} (q) ;
\draw[->]  (0) to[loop right]  node [midway,anchor=west] {$\mathsf{e}$} (0);
\draw[->]  (1) to [loop right] node [midway,anchor=west] {$\mathsf{d}$} (1) ;

\end{tikzpicture}
\caption{An example of an $\SG$ with $\states = \set{\mathsf{p},\mathsf{q},\target,\sink}$, $\states<\Box> = \set{\mathsf{q},\target}$, $\states<\circ> = \set{\mathsf{p},\sink}$, the initial state $\mathsf{p}$ and the set of actions $\actions = \set{\mathsf{a},\mathsf{b},\mathsf{c},\mathsf{d},\mathsf{e}}$; $\Av(\mathsf{p})=\set{\mathsf{a}}$ with $\trans(\mathsf{p},\mathsf{a})(\mathsf{q})=1$;  $\Av(\mathsf{q})=\set{\mathsf{b},\mathsf{c}}$ with $\trans(\mathsf{q},\mathsf{b})(\mathsf{p})=1$ and $\trans(\mathsf{q},\mathsf{c})(\mathsf{q}) = \trans(\mathsf{q},\mathsf{c})(\target) = \trans(\mathsf{q},\mathsf{c})(\sink) = \frac{1}{3}$. 
	For actions with only one successor, we do not depict the transition probability $1$.
}
\label{SGex}
\myspace
\end{figure}

%% file: 3_example.tex
\vspace{-1em}

\section{Example}\label{sec:example}

\vspace{-0.8em}

In this section, we illustrate the issues preventing BVI convergence and our solution on a few examples.
Recall that $\gub$ is the sequence converging to the greatest solution of the Bellman equations, while $\ub$ is in general any sequence over-approximating $\val$ that one or another BVI algorithm suggests.

Firstly, we illustrate the issue that arises already for the special case of MDP. 
Consider the MPD of Figure \ref{lBCEC} on the left. 
Although $\val(\mathsf s)=\val(\mathsf t)=0.5$, we have $\gub<i>(\mathsf s)=\gub<i>(\mathsf t)=1$ for all $i$.
Indeed, the upper bound for $\mathsf t$ is always updated as the maximum of $\gub<i>(\mathsf t,\mathsf c)$ and $\gub<i>(\mathsf t,\mathsf b)$.
Although $\gub<i>(\mathsf t,\mathsf c)$ decreases over time, $\gub<i>(\mathsf t,\mathsf b)$ remains the same, namely equal to $\gub<i>(\mathsf s)$, which in turn remains equal to $\gub<i>(\mathsf s,\mathsf a)=\gub<i>(\mathsf t)$.
This cyclic dependency lets both $\mathsf s$ and $\mathsf t$ remain in an ``illusion'' that the value of the other one is $1$.

The solution for MDP is to remove this cyclic dependency by collapsing all MECs into singletons and removing the resulting purely self-looping actions.
Figure \ref{lBCEC} in the middle shows the MDP after collapsing the EC $\set{\mathsf{s},\mathsf{t}}$. 
This turns the MDP into a stopping one, where $\target$ or $\sink$ is under any strategy reached with probability 1.
In such MDP, there is a unique solution to the Bellman equations.
Therefore, the greatest fixpoint is equal to the least one and thus to $\val$.

\noindent
\begin{figure}[t]
	\myspace\myspace
	
	\begin{minipage}[t]{0.36 \textwidth}
		\input{littleBCEC}
	\end{minipage}
	\begin{minipage}[t]{0.36\textwidth}
		\input{littleBCEC_coll}
	\end{minipage}
	\begin{minipage}[t]{0.25 \textwidth}  
		\vspace*{-9em}
		\large
		\begin{center}
			\normalsize
			\begin{tabular}{|c||c|c|c|c|} \hline
				$i$  & $\lb<i>(\mathsf{\{s,t\}})$ &  $\gub<i>(\mathsf{\{s,t\}})$ \\ \hline \hline
				0 		&	0			&		1			\\ \hline
				1 		&	$\frac{1}{3}$	&		$\frac{2}{3}$ 	\\ \hline
				2 		&	$\frac{4}{9}$	&		$\frac{5}{9}$	\\ \hline
				3  		&	$\frac{13}{27}$	&		$\frac{14}{27}$	 \\ \hline
			\end{tabular}
		\end{center}
	\end{minipage}
\myspace
	\caption{Left: An MDP (as special case of SG) where BVI does not converge due to the
          grayed EC. Middle: The same MDP where the EC is collapsed,
          making BVI converge. Right: The approximations illustrating
          the convergence of the MDP in the middle.}
	\label{lBCEC}
\end{figure}
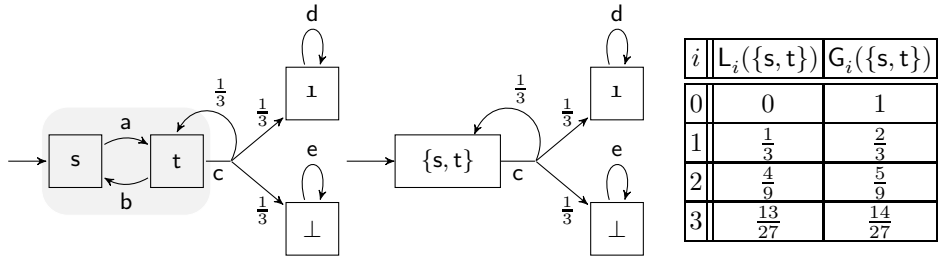




Secondly, we illustrate the issues that additionally arise for general \SG.
It turns out that the collapsing approach can be extended only to games where all states of each EC belong to one player only \cite{MatPhD}.
In this case, both Maximizer's and Minimizer's ECs are collapsed the same way as in MDP. 



However, when both players are present in an EC, then collapsing may not solve the issue.
Consider the SG of Figure \ref{cBCEC}.
Here $\alpha$ and $\beta$ represent the values of the respective actions.\footnote{Precisely, we consider them to stand for a probabilistic branching with probability $\alpha$ (or $\beta$) to $\target$ and with the remaining probability to $\sink$. 
To avoid clutter in the figure, we omit this branching and depict only the value.}
There are three cases:

First, let $\mathsf{\alpha}<\mathsf{\beta}$.
If the bounds converge to these values we eventually observe $\gub<i>(q,e)<\lb<i>(r,f)$ and learn the induced inequality.
Since $\mathsf{p}$ is a Minimizer's state it will never pick the action leading to the greater value of $\mathsf{\beta}$.
Therefore, we can safely merge $\mathsf{p}$ and $\mathsf{q}$, and remove the action leading to $\mathsf{r}$, as shown in the second subfigure.

Second, if $\mathsf{\alpha}>\mathsf{\beta}$, $\mathsf{p}$ and $\mathsf{r}$ can be merged in an analogous way, as shown in the third subfigure.

Third, if $\mathsf{\alpha}=\mathsf{\beta}$, both previous solutions as well as collapsing all three states as in the fourth subfigure is possible.
However, since the approximants may only converge to $\mathsf{\alpha}$ and $\mathsf{\beta}$ in the limit, we may not know in finite time which of these cases applies and thus cannot decide for any of the collapses.

\begin{figure}[htbp]

\noindent
\begin{minipage}[t]{0.24 \textwidth}
		\input{complexBCEC}
\end{minipage}
\begin{minipage}[t]{0.24 \textwidth}
		\input{complexBCEC_alpha}
\end{minipage}
\begin{minipage}[t]{0.24 \textwidth}
		\input{complexBCEC_beta}
\end{minipage}
\begin{minipage}[t]{0.24 \textwidth}
		\input{complexBCEC_both}
\end{minipage}
\caption{Left: 
	Collapsing ECs in SG may lead to incorrect results.
	The Greek letters on the leaving arrows denote the values of the \leaving\ actions. Right three figures: Correct collapsing in different cases, depending on the relationship of $\alpha$ and $\beta$.
	In contrast to MDP, some actions of the EC \leaving\ the collapsed part have to be removed. 
	}
\label{cBCEC}
\myspace\myspaceb
\end{figure}
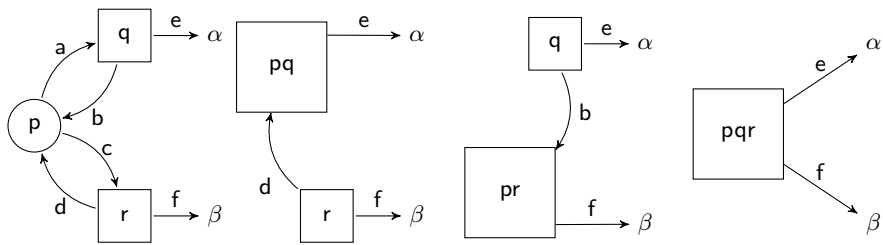

Consequently, the approach of collapsing is not applicable in general.
In order to ensure BVI convergence, we suggest a different method, which we call \emph{deflating}. 
It does not involve changing the state space, but rather decreasing the upper bound $\ub<i>$ to the least value that is currently provable (and thus still correct).
To this end, we analyze the \leaving\ actions, i.e. with successors outside of the EC, for the following reason. 
If the play stays in the EC forever, the target is never reached and Minimizer wins.
Therefore, Maximizer needs to pick some \leaving\ action to avoid staying in the EC.

For the EC with the states $\mathsf{s}$ and $\mathsf{t}$  in Figure \ref{lBCEC}, the only \leaving\ action is $\mathsf{c}$. In this example, since $\mathsf{c}$ is the only \leaving\ action, $\ub<i>(\mathsf{t},\mathsf{c})$ is the highest possible upper bound that the EC can achieve. Thus, by decreasing the upper bound of all states in the EC to that number\footnote{We choose the name ``deflating'' to evoke decreasing the overly high ``pressure'' in the EC until it equalizes with the actual ``pressure'' outside.}, we still have a safe upper bound.
Moreover, with this modification BVI converges in this example, intuitively because now the upper bound of $\mathsf{t}$ depends on action $\mathsf{c}$ as it should. 

For the example in Figure \ref{cBCEC}, it is correct to decrease the upper bound to the maximal \leaving\ one, i.e. 
$\max \{\hat{\alpha},\hat{\beta}\}$, where $\hat\alpha:=\ub<i>(\mathsf{a}),\hat{\beta}:=\ub<i>(\mathsf{b})$ are the current approximations of $\alpha$ and of $\beta$.
However, this itself does not ensure BVI convergence.
Indeed, if for instance $\hat{\alpha} < \hat{\beta}$ then deflating all states to $\hat\beta$ is not tight enough, as values of $\mathsf{p}$ and $\mathsf{q}$ can even be bounded by  $\hat{\alpha}$. 
In fact, we have to find a certain sub-EC that corresponds to $\hat{\alpha}$, in this case $\{\mathsf{p,q}\}$ and set all its upper bounds to $\hat{\alpha}$.
We define and compute these sub-ECs in the next section.

In summary, the general structure of our convergent BVI algorithm is to produce the sequence $\ub$ by application of Bellman updates and occasionally find the relevant sub-ECs and deflate them. 
The main technical challenge is that states in an EC in SG can have different values, in contrast to the case of MDP.

%% file: littleBCEC.tex

\begin{tikzpicture}[scale=0.9]
\draw[white,fill=black!5, rounded corners=10pt] (-0.5,0.8) rectangle (2,-0.8);

\drawdummy (init) at (-1,0) {};
\drawbox (p) at (0,0) {$\mathsf{s}$};
\drawbox (q) at (1.5,0) {$\mathsf{t}$};
\drawdummy (mid) at (2.3,0) {};
\drawbox (1) at (3.5,1) {$\target$ };
\drawbox (0) at (3.5,-1)  {$\bot$};

\draw[->] (init) to (p);
\draw[->]  (p) to[bend left] node [midway,anchor=south] {$\mathsf{a}$}(q) ;
\draw[->]  (q) to [bend left] node [midway,anchor=north] {$\mathsf{b}$} (p);
\draw[-] (q) to node [midway,anchor=north] {$\mathsf{c}$} (mid) ;
\draw (mid) -- node [anchor=north,pos=0.6] {$\frac{1}{3}$} (0);
\draw (mid) -- node [anchor=south,pos=0.6] {$\frac{1}{3}$} (1);
\draw (mid) to [bend left,looseness=2,out=270,in=270] node [anchor=south,pos=0.5] {$\frac{1}{3}$} (q) ;
\draw[->]  (0) to[loop above]  node [midway,anchor=south] {$\mathsf{e}$} (0);
\draw[->]  (1) to [loop above] node [midway,anchor=south] {$\mathsf{d}$} (1) ;
\end{tikzpicture}

%% file: littleBCEC_coll.tex

\begin{tikzpicture}[scale=0.9]
\drawdummy (init) at (-1,0) {};
\node[draw,rectangle,minimum height=0.7cm, minimum width=1.4cm] (pq) at (0.5,0) {$\mathsf{ \{s,t\}}$};
\drawdummy (mid) at (1.8,0) {};
\drawbox (1) at (3,1) {$\target$ };
\drawbox (0) at (3,-1)  {$\bot$};

\draw[->] (init) to (pq);
\draw[-] (pq) to node [midway,anchor=north] {$\mathsf{c}$} (mid) ;
\draw (mid) -- node [anchor=north,pos=0.6] {$\frac{1}{3}$} (0);
\draw (mid) -- node [anchor=south,pos=0.6] {$\frac{1}{3}$} (1);
\draw (mid) to [bend left,looseness=2,out=270,in=270] node [anchor=south,pos=0.5] {$\frac{1}{3}$}($(pq.north east)!0.5!(pq.north)$);
\draw[->]  (0) to[loop above]  node [midway,anchor=south] {$\mathsf{e}$} (0);
\draw[->]  (1) to [loop above] node [midway,anchor=south] {$\mathsf{d}$} (1) ;

\end{tikzpicture}

%% file: complexBCEC.tex

\begin{tikzpicture}[scale=0.6]

\drawcirc (A) at (0,0) {$\mathsf{p}$};
\drawbox (B) at (2,2) {$\mathsf{q}$};
\drawbox (C) at (2,-2) {$\mathsf{r}$};

\node (Bo) at (4,2) {$\alpha$};
\node (Co) at (4,-2) {$\beta$};

\draw[->] (A) to[bend left]  node [midway,anchor=south] {$\mathsf{a}$} (B);
\draw[->] (B) to[bend left]  node [midway,anchor=north] {$\mathsf{b}$} (A);
\draw[->] (A) to[bend left]  node [midway,anchor=south, pos=0.7] {$\mathsf{c}$} (C);
\draw[->] (C) to[bend left]  node [midway,anchor=north] {$\mathsf{d}$} (A);

\draw[->] (B) to  node [midway,anchor=south] {$\mathsf{e}$} (Bo);
\draw[->] (C) to  node [midway,anchor=south] {$\mathsf{f}$} (Co);

\end{tikzpicture}


%% file: complexBCEC_alpha.tex

\begin{tikzpicture}[scale=0.6]

\node[draw,rectangle,minimum size=1.2cm] (AB) at (1,1.3) {$\mathsf{pq}$};

\drawbox (C) at (2,-2) {$\mathsf{r}$};

\node (Bo) at (4,2) {$\alpha$};
\node (Co) at (4,-2) {$\beta$};

\draw[->] (C) to[bend left]  node [midway,anchor=north east, pos=0.3] {$\mathsf{d}$} (AB);

\draw[->] (AB.35) to  node [midway,anchor=south] {$\mathsf{e}$} (Bo);
\draw[->] (C) to  node [midway,anchor=south] {$\mathsf{f}$} (Co);

\end{tikzpicture}


%% file: complexBCEC_beta.tex

\begin{tikzpicture}[scale=0.6]

\node[draw,rectangle,minimum size=1.2cm] (AC) at (1,-1.3) {$\mathsf{pr}$};
\drawbox (B) at (2,2) {$\mathsf{q}$};

\node (Bo) at (4,2) {$\alpha$};
\node (Co) at (4,-2) {$\beta$};

\draw[->] (B) to[bend left]  node [midway,anchor=west] {$\mathsf{b}$} (AC);

\draw[->] (B) to  node [midway,anchor=south] {$\mathsf{e}$} (Bo);
\draw[->] (AC.-35) to  node [midway,anchor=south] {$\mathsf{f}$} (Co);

\end{tikzpicture}


%% file: complexBCEC_both.tex

\begin{tikzpicture}[scale=0.6]

\node[draw,rectangle,minimum size=1.2cm] (ABC) at (1,0) {$\mathsf{pqr}$};

\node (Bo) at (4,2) {$\alpha$};
\node (Co) at (4,-2) {$\beta$};

\draw[->] (ABC) to  node [midway,anchor=south] {$\mathsf{e}$} (Bo);
\draw[->] (ABC) to node [midway,anchor=south] {$\mathsf{f}$} (Co);

\end{tikzpicture}


%% file: 4_main.tex
\vspace{-1.2em}

\section{Convergent Over-approximation} \label{sec:main}

\vspace{-0.8em}

In Section \ref{sec:char}, we characterize \SG s where Bellman equations have more solutions.
Based on the analysis, subsequent sections show how to alter the procedure computing the sequence $\gub<i>$ over-approximating $\val$ so that the resulting tighter sequence $\ub<i>$ still over-approximates $\val$, but also converges to $\val$.
This ensures that thus modified BVI converges.
Section~\ref{sec:learn} presents the learning-based variant of our BVI.
\vspace{-1.3em}

\subsection{Bloated end components cause non-convergence}\label{sec:char}
\vspace{-0.5em}

As we have seen in the example of Fig.~\ref{cBCEC}, BVI generally does not converge due to ECs with a particular structure of the \leaving\ actions.
The analysis of ECs relies on the extremal values that can be achieved by \leaving\ actions (in the example, $\alpha$ and $\beta$).
Given the value function $\val$ or just its current over-approximation $\ub<i>$, we define the most profitable \leaving\ action for Maximizer (denoted by $\Box$) and Minimizer (denoted by $\circ$) as follows.
\begin{definition}[$\exit$]\label{def:exit}
Given a set of states $T \subseteq \states$ and a function $f:\states\to[0,1]$ (see footnote \ref{fn:overload}),
the $f$-value of the best $T$-\leaving\ action of Maximizer and Minimizer, respectively, is defined as
\vspace{-0.5em}
\begin{align*}
	\exit<f>[\Box](T) &= \max_{\substack{\state \in T_\Box\\ (\state,\action) \leaves T}} f(\state,\action)\\
	\exit<f>[\circ](T) &= \min_{\substack{\state \in T_\circ\\ (\state,\action) \leaves T}} f(\state,\action)
\end{align*}
with the convention that $\max_\emptyset = 0$ and $\min_\emptyset = 1$.
\end{definition}

\begin{example}
	In the example of Fig.~\ref{cBCEC} on the left with $T=\set{\mathsf p, \mathsf q,\mathsf r}$ and $\alpha<\beta$, we have $	\exit<\val>[\Box](T)=\beta$, 	$\exit<\val>[\circ](T)=1$.
	It is due to $\beta<1$ that BVI does not converge here.
	We generalize this in the following lemma.
	\qee
\end{example}


\begin{lemma}\label{lem:illu} 
Let $T$ be an EC.
For every $m$ satisfying
$\exit<\val>[\Box](T) \leq m \leq \exit<\val>[\circ](T)$,
there is a solution $f\colon\states\to[0,1]$ to the Bellman equations, which on $T$ is constant and equal to $m$.
\end{lemma}
\begin{proof}[Idea]
Intuitively, such a constant $m$ is a solution to the Bellman equations on $T$ for the following reasons.
As both players prefer getting $m$ to \leaving\ and getting ``only'' the values of their respective $\exit$,
they both choose to stay in the EC (and the extrema in the Bellman equations are realized on non-\leaving\ actions). 
On the one hand, Maximizer (Bellman equations with $\max$) is hoping for the promised $m$, which is however not backed up by any actions actually \leaving\ towards the target.
On the other hand, Minimizer (Bellman equations with $\min$) does not realize that staying forever results in her optimal value $0$ instead of $m$.
\qed
\end{proof}

\begin{corollary}
	If 	$\exit<\val>[\circ](T) > \exit<\val>[\Box](T)$ for some EC~$T$,
	then $\gub[*]\neq\val$.
\end{corollary}
\begin{proof}
	Since there are $m_1,m_2$ such that $\exit<\val>[\Box](T) < m_1 < m_2 < \exit<\val>[\circ](T)$, by Lemma~\ref{lem:illu} there are two different solutions to the Bellman equations.
	In particular, $\gub[*]>\lb[*]=\val$, and BVI does not converge.
	\qed
\end{proof}


In accordance with our intuition that ECs satisfying the above inequality should be deflated, we call them bloated.
\begin{definition}[$\CEC$]
	An EC $T$ is called a \emph{bloated end component ($\CEC$)}, if $\exit<\val>[\circ](T) > \exit<\val>[\Box](T)$.
\end{definition}

\begin{example}
	In the example of Fig.~\ref{cBCEC} on the left with $\alpha<\beta$, 
	the ECs $\mathsf{\{p,q\}}$ and $\mathsf{\{p,q,r\}}$ are $\CEC$s. \qee
\end{example}

\begin{example}
If an EC $T$ has no \leaving\ actions of Minimizer (or no Minimizer's states at all, as in an MDP), then $\exit<\val>[\circ](T)=1$ (the case with $\min_\emptyset$).
Hence all numbers between $\exit<\val>[\Box](T)$ and 1 are a solution to the Bellman equations and $\gub[*](s)=1$ for all states $s\in T$.

Analogously, if Maximizer does not have any \leaving\ action in $T$
, then $\exit<\val>[\Box](T)=0$ (the case with $\max_\emptyset$), it is a $\CEC$ and all numbers between 0 and $\exit<\val>[\circ](T)$ are a solution to the Bellman equations. 

Note that in MDP all ECs belong to one player, namely Maximizer.
Consequently, all ECs are $\CEC$s except for ECs where Maximizer has an \leaving\ action with value $1$;
all other ECs thus have to be collapsed (or deflated) to ensure BVI convergence in MDPs.
Interestingly, all non-trivial ECs in MDPs are a problem, while in SGs through the presence of the other player some ECs can converge, namely if both players want to exit. Such an EC is depicted in Appendix \ref{app:convEC}
\qee
\end{example}

We show that $\CEC$s are indeed the only obstacle for BVI convergence.
\begin{theorem}
	\label{lem:convBVIwo}
	If the $\SG$ contains no $\CEC$s except for $\{\mathfrak 0\}$ and $\{\mathfrak 1\}$, then
	$\gub[*]=\val$.
\end{theorem}

\begin{proof}[Sketch]
	Assume, towards a contradiction, that there is some state $\state$ with a positive difference $\gub[*](\state)-\val(\state)>0$.
	Consider the set $D$ of states with the maximal difference. 
	$D$ can be shown to be an EC.
	Since it is not a $\CEC$ there has to be an action \leaving\ $D$ and realizing the optimum in that state.
	Consequently, this action also has the maximal difference, and all its successors, too.
	Since some of the successors are outside of $D$, we get a contradiction with the maximality of $D$.
	\qed
\end{proof}

\vspace{-0.5em}
In Section \ref{sec:static}, we show how to eliminate $\CEC$s by collapsing their ``core'' parts, called below MSECs (maximal simple end components). 
Since MSECs can only be identified with enough information about $\val$, Section~\ref{sec:dynamic} shows how to avoid direct {\em a priori} collapsing and instead dynamically deflate candidates for MSECs in a conservative way.

\vspace{-1.3em}

\subsection{Static MSEC decomposition}\label{sec:static}

\vspace{-0.8em}

Now we turn our attention to SG with $\CEC$s.
Intuitively, since in a $\CEC$ all Minimizer's \leaving\ actions have a higher value than what Maximizer can achieve, Minimizer does not want to use any of his own \leaving\ actions and prefers staying in the EC (or steering Maximizer towards his worse \leaving\ actions).
Consequently, only Maximizer wants to take an \leaving\ action.
In the MDP case he can pick any desirable one.
Indeed, he can wait until he reaches a state where it is available.
As a result, in MDP all states of an EC have the \emph{same value} and can all be collapsed into one state.
In the SG case, he may be restricted by Minimizer's behaviour or even not given any chance to exit the EC at all.
As a result, a $\CEC$ may contain several parts (below denoted MSECs), each with different value, intuitively corresponding to different exits. 
Thus instead of MECs, we have to decompose into finer MSECs and only collapse these.
\myspaceb

\begin{definition}[$\Simple~EC$]\label{def:simBCEC}
	An EC~$T$ is called \emph{\simple\ (SEC)}, if for all $\state \in T$ we have $\val(\state)=\exit<\val>[\Box](T)$.

	A SEC $C$ is \emph{maximal (MSEC)} if there is no SEC $C'$ such that $C\subsetneq C'$. 
\end{definition}
\myspaceb

Intuitively, an EC is \simple, if Minimizer cannot keep Maximizer away from his $\exit$.
Independently of Minimizer's decisions, Maximizer can reach the $\exit$ almost surely, unless Minimizer decides to leave, in which case Maximizer could achieve an even higher value.
\myspaceb

\begin{example}
	Assume $\alpha<\beta$ in the example of Figure \ref{cBCEC}. 
	Then $\set{\mathsf p, \mathsf q}$ is a SEC and an MSEC.
	Further observe that action $\mathsf c$ is sub-optimal for Minimizer and removing it does not affect the value of any state, but simplifies the graph structure.
	Namely, it destructs the whole EC into several (here only one) SECs and some non-EC states (here $\mathsf r$).
	\qee
\end{example}
\myspaceb

Algorithm~\ref{alg:FIND}, called $\FIND$, shows how to compute MSECs.
It returns the set of all MSECs if called with parameter $\val$.
However, later we also call this function with other parameters $f:\states\to[0,1]$.
The idea of the algorithm is the following.
The set $X$ consists of Minimizer's sub-optimal actions, leading to a higher value.
As such they cannot be a part of any SEC and thus should be ignored  when identifying SECs.
(The previous example illustrates that ignoring $X$ is indeed safe as it does not change the value of the game.)
We denote the game $\G$ where the available actions $\Av$ are changed to the new available actions $\Av'$ (ignoring the Minimizer's sub-optimal ones) as $\G<[\Av / \Av']>$.
Once removed, Minimizer has no choices to affect the value and thus each EC is simple.

\vspace{-1.8em}

\begin{algorithm}[H]

\caption{$\FIND$}\label{alg:FIND}
\begin{algorithmic}[1]
\Function{$\FIND$}{$f:\states\to[0,1]$}
\State $X \gets \{(\state,\{\action\in\Av(\state) \mid f(\state,\action) > f(\state)\}) \mid \state \in \states<\circ>\}$\label{find:remove}
\State 
$\Av'\gets\Av\setminus X$ 
\hspace*{14mm}\Comment{Minimizer's $f$-suboptimal actions removed} \label{line:findRemove}
\State \textbf{return} $\MECs(\G<[\Av / \Av']>)$
\hspace*{2mm}\Comment{$\MECs(\G<[\Av / \Av']>)$ are MSECs of the original $\G$}
\EndFunction
\end{algorithmic}

\end{algorithm}

\vspace{-1.8em}

%

\begin{lemma}[Correctness of Algorithm \ref{alg:FIND}]\label{lem:find}
	  $T \in$ \emph{\FIND}$(\val)$ if and only if $T$ is a MSEC.
\end{lemma}
\vspace{-0.8em}
\begin{proof}[Sketch]
``If'': Since $T$ is an MSEC, all states in $T$ have the value $\exit<\val>[\Box](T)$, and hence also all actions that stay inside $T$ have this value. Thus, no action that stays in $T$ is removed by Line \ref{line:findRemove} and it is still a MEC in the modified game.

``Only if'': If $T \in\FIND(\val)$, then $T$ is a MEC of the game where the suboptimal available actions (those in $X$) of Minimizer have been removed. Hence for all $\state \in T: \val(\state) = \exit<\val>[\Box](T)$, because intuitively Minimizer has no possibility to influence the value any further, since all actions that could do so were in $X$ and have been removed. Since $T$ is a MEC in the modified game, it certainly is an EC in the original game. Hence $T$ is a SEC. The inclusion maximality follows from the fact that we compute MECs in the modified game. Thus $T$ is an MSEC.
\qed
\end{proof}

\vspace{-0.8em}
\begin{remark}[Algorithm with an oracle]
	In Section \ref{sec:example}, we have seen that collapsing MECs does not ensure BVI convergence.
	Collapsing does not preserve the values, since in BECs we would be collapsing states with different values.
	Hence we want to collapse only MSECs, where the values are the same.
	If, moreover, we remove $X$ in such a collapsed SG, then there are no (non-sink) ECs and BVI converges on this SG to the original value.	
\end{remark}
	
\vspace{-0.8em}
	The difficulty with this algorithm is that it requires an oracle to compare values, for instance a sufficiently precise approximation of $\val$. 
	Consequently, we cannot pre-compute the MSECs, but have to find them while running BVI. 
	Moreover, since the approximations converge only in the limit we may never be able to conclude on simplicity of some ECs. For instance, if $\alpha=\beta$ in Figure \ref{cBCEC}, and if the approximations converge at different speeds, then Algorithm~\ref{alg:FIND} always outputs only a part of the EC, although the whole EC on $\mathsf{\{p,q,r\}}$ is \simple.
	
	In MDPs, all ECs are \simple, because there is no second player to be resolved and all states in an EC have the same value. Thus for MDPs it suffices to collapse all MECs, in contrast to SG. 

\vspace*{-1em}

\subsection{Dynamic MSEC decomposition}\label{sec:dynamic}

\vspace*{-0.5em}

Since MSECs cannot be identified from approximants of $\val$ for sure, we refrain from collapsing\footnote{Our subsequent method can be combined with local collapsing whenever the lower and upper bounds on $\val$ are conclusive.} and instead only decrease the over-approximation in the corresponding way.
We call the method \emph{deflating}, by which we mean decreasing the upper bound of all states in an EC to its $\exit<\ub>[\Box]$, see Algorithm \ref{alg:adjust}. 
The procedure $\ADJUST$ (called on the current upper bound $\ub<i>$) decreases this upper bound to the minimum possible value according to the current approximation and thus prevents states from only depending on each other, as in SECs.
Intuitively, it gradually approximates SECs and performs the corresponding adjustments, but does not commit to any of the approximations.

\vspace*{-2em}

\begin{algorithm}[H]
\caption{$\ADJUST$}\label{alg:adjust}
\begin{algorithmic}[1]
\Function{$\ADJUST$}{EC $T$, $f:\states\to[0,1]$}
	\For {$\state \in T$} 
		\State $\displaystyle f(\state) \gets \min(f(\state), \exit<f>[\Box](T))$\label{line:adjust} \Comment{Decrease the upper bound}
	\EndFor
	\State \Return $f$
\EndFunction
\end{algorithmic}
\end{algorithm}

\vspace*{-2.8em}
%
\begin{lemma}[$\ADJUST$ is sound]
	\label{lem:corrADJ}
For any $f:\states\to[0,1]$ such that $f\geq\val$ and any EC $T$,
$\ADJUST(T,f)\geq\val$.
\end{lemma}

\vspace*{-0.5em}
This allows us to define our BVI algorithm as the naive BVI with only the additional lines \ref{line:loopAdjust}-\ref{line:callAdjust}, see Algorithm \ref{alg:BVIA}.

\vspace*{-2em}

\begin{algorithm}[H]
\caption{$\UPDATE$ procedure for bounded value iteration on SG}\label{alg:BVIA}
\begin{algorithmic}[1]
\Procedure{$\UPDATE$}{$\lb:\states\to[0,1]$, $\ub:\states\to[0,1]$}
\State $\lb,\ub$ get updated according to Eq.~(\ref{eq:Lsa})~and~(\ref{eq:Ls}) \Comment{Bellman updates}
\smallskip
\For{$T \in \FIND(\lb)$} \label{line:loopAdjust} ~~~~~~~~\Comment{Use lower bound to find ECs}
	\State $\ub\gets\ADJUST(T,\ub)$ \label{line:callAdjust} ~~~~~~~~~~~~~\Comment{and deflate the upper bound there}
\EndFor
\EndProcedure
\end{algorithmic}
\end{algorithm}

\vspace*{-2.3em}

\begin{theorem}[Soundness and completeness]\label{BVIA_term}
Algorithm~\ref{alg:gen} (calling Algorithm~\ref{alg:BVIA}) produces monotonic sequences $\lb$ under- and $\ub$ over-approximating $\val$, and terminates.
\end{theorem}
\vspace*{-1.2em}
\begin{proof}[Sketch]
	The crux is to show that $\ub$ converges to $\val$.
	We assume towards a contradiction, that there exists a state $\state$ with $\lim_{n \to \infty} \ub<n>(\state)-\val(\state) > 0$. Then there exists a nonempty set of states $X$ where the difference between $\lim_{n \to \infty} \ub<n>$ and $\val$ is maximal. If the upper bound of states in $X$ depends on states outside of $X$, this yields a contradiction, because then the difference between upper bound and value would decrease in the next Bellman update. So $X$ must be an EC where all states depend on each other. However, if that is the case, calling $\ADJUST$ decreases the upper bound to something depending on the states outside of $X$, thus also yielding a contradiction.
\qed 
\end{proof}

\vspace*{-1.8em}
\subsection*{Summary of our approach:}
\vspace*{-0.5em}
\begin{enumerate}
	\item We cannot collapse MECs, because we cannot collapse BECs with non-constant values. 
	\item If we remove $X$ (the sub-optimal actions of Minimizer) we can collapse MECs (now actually MSECs with constant values).
	\item Since we know neither $X$ nor SECs we gradually deflate SEC approximations. 
\end{enumerate}

\subsection{Learning-based algorithm}\label{sec:learn}

\emph{Asynchronous value iteration} selects in each round a subset $T\subseteq S$ of states and performs the Bellman update in that round only on $T$.
Consequently, it may speed up computation if ``important'' states are selected.
However, using the standard VI it is even more difficult to determine the current error bound.
Moreover, if some states are not selected infinitely often the lower bound may not even converge.

In the setting of bounded value iteration, the current error bound is known for each state and thus convergence can easily be enforced.
This gave rise to asynchronous VI, such as BRTDP (bounded real time dynamic programing) in the setting of stopping MDPs \cite{BRTDP}, where the states are selected as those that appear on a simulation run.
Very similar is the adaptation for general MDP \cite{atva}.
In order to simulate a run, the transition probabilities determine how to resolve the probabilistic choice.
In order to resolve the non-deterministic choice of Maximizer, the ``most promising action'' is taken, i.e., with the highest $\ub$.
This choice is derived from a reinforcement algorithm called delayed Q-learning and ensures convergence while practically performing well~\cite{atva}.

In this section, we harvest our convergence results and BVI algorithm for SG, which allow us to trivially extend the asynchronous learning-based approach of BRTDP to SGs.
On the one hand, the only difference to the MDP algorithm is how to resolve the choice for Minimizer.
Since the situation is dual, we again pick the ``most promising action'', in this case with the lowest $\lb$.
On the other hand, the only difference to Algorithm~\ref{alg:gen} calling Algorithm~\ref{alg:BVIA} is that the Bellman updates of $\ub$ and $\lb$  are performed on the states of the simulation run only, see lines \ref{line:sim}-\ref{line:sim-up} of Algorithm~\ref{alg:BRTDP}. 

\vspace{-2em}

\begin{algorithm}[H]
\caption{Update procedure for the learning/BRTDP version of BVI on SG}\label{alg:BRTDP}
\begin{algorithmic}[1]
\Procedure{$\UPDATE$}{$\lb:\states\to[0,1]$, $\ub:\states\to[0,1]$}
\State $\rho\gets$ path $\initstate,\state<1>,\ldots,\state<\ell>$ of length $\ell\leq k$, obtained by simulation where the successor of $\state$ is $\state'$ with probability $\trans(\state,\action,\state')$ and $\action$ is sampled randomly from 
$\arg\max_{\action} \ub(\state,\action)$ and $\arg\min_{\action} \lb(\state,\action)$ for $\state\in\states<\Box>$ and $\state\in\states<\circ>$, respectively \label{line:sim}
\State $\lb,\ub$ get updated by Eq.~(\ref{eq:Lsa}) and (\ref{eq:Ls}) on states $\state<\ell>,\state<\ell-1>,\ldots,\initstate$\hspace*{-2mm} \Comment{all $\state\in\rho$} \label{line:sim-up}
\smallskip
\For{$T \in \FIND(\lb)$} \label{line:loopAdjust2}
	\State $\ADJUST(T,\ub)$ \label{line:callAdjust2}
\EndFor
\EndProcedure
\end{algorithmic}
\end{algorithm}

\vspace{-2em}

If $\target$ or $\sink$ is reached in a simulation, we can terminate it.
It can happen that the simulation cycles in an EC. 
To that end, we have a bound $k$ on the maximum number of steps.
The choice of $k$ is discussed in detail in \cite{atva} and we use~$2\cdot |\states|$ to guarantee the possibility of reaching sinks as well as exploring new states.
If the simulation cycles in an EC, the subsequent call of $\ADJUST$ ensures that next time there is a positive probability to exit this EC.
Further details can be found in Appendix~\ref{app:brtdp}.
\vspace{-0.5em}

%
%


%% file: 5_ImplAndRes.tex
\section{Experimental results}\label{sec:4}\label{sec:exper}

\vspace{-1em}
We implemented both our algorithms as an extension of PRISM-games~\cite{PRISM-games}, a branch of PRISM~\cite{prism} that allows for modelling $\SGs$, utilizing previous work of \cite{atva,MatPhD} for MDP and SG with single-player ECs.
We tested the implementation on the SGs from the PRISM-games case studies~\cite{gamesBen} that have reachability properties and one additional model from~\cite{cloud} that was also used in~\cite{MatPhD}. We compared the results with both the explicit and the hybrid engine of PRISM-games, but since the models are small both of them performed similar and we only display the results of the hybrid engine in Table~\ref{table:time}.

Furthermore we ran experiments on MDPs from the PRISM benchmark suite~\cite{PRISMben}.
We compared our results there to the hybrid and explicit engine of PRISM, the interval iteration implemented in PRISM~\cite{BVI}, the hybrid engine of \textsc{Storm}~\cite{Storm} and the BRTDP implementation of~\cite{atva}.

Recall that the aim of the paper is not to provide a faster VI algorithm, but rather the first  guaranteed one.
Consequently, the aim of the experiments is not to show any speed ups, but to experimentally estimate the overhead needed for computing the guarantees.

The appendix contains information on the technical details of the experiments (\ref{app:tech}), all the models (\ref{sec:models}) and the tables for the experiments on MDPs (\ref{sec:MDPexp}). Note that although some of the SG models are parametrized they could only be scaled by manually changing the model file, which complicates extensive benchmarking.

Although our approaches compute the additional upper bound to give the convergence guarantees, for each of the experiments one of our algorithms performed similar to PRISM-games. Table~\ref{table:time} shows this result for three of the four SG models in the benchmarking set. On the fourth model, PRISM's pre-computations already solve the problem and hence it cannot be used to compare the approaches. For completeness, the results are displayed in Appendix \ref{sec:SGexp}.

Whenever there are few MSECs, as in mdsm and cdmsn, BVI performs like PRISM-games, because only little time is used for deflating. 
Apparently the additional upper bound computation takes very little time in comparison to the other tasks (e.g. parsing, generating the model, pre-computation) and does not slow down the verification significantly.
For cloud, BVI is slower than PRISM-games, because there are thousands of MSECs and deflating them takes over 80\% of the time. This comes from the fact that we need to compute the expensive end component decomposition for each deflating step.
BRTDP performs well for cloud, because in this model, as well as generally often if there are many MECs~\cite{atva}, only a small part of the state space is relevant for convergence. For the other models, BRTDP is slower than the deterministic approaches, because the models are so small that it is faster to first construct them completely than to explore them by simulation. 

 
\input{time}

Our more extensive experiments on MDPs compare the guaranteed approaches based on collapsing (learning-based from~\cite{atva} and deterministic from \cite{BVI}) to our guaranteed approaches based on deflating (so BRTDP and BVI).
Since both learning-based approaches as well as both deterministic approaches perform similarly (see Table \ref{table:timeMDP} in Appendix \ref{sec:MDPexp}), we conclude that collapsing and deflating are both useful for practical purposes, while the latter is also applicable to SGs.
Furthermore we compared the usual unguaranteed value iteration of PRISM's explicit engine to BVI and saw that our guaranteed approach did not take significantly more time in most cases. This strengthens the point that the overhead for the computation of the guarantees is negligible
\vspace{-1.3em}

%% file: time.tex
\begin{table}[t] \centering \caption{Experimental results for the experiments on SGs. The left two columns denote the model and the given parameters, if present. Columns 3 to 5 display the verification time in seconds for each of the solvers, namely PRISM-games (referred as PRISM), our BVI algorithm (BVI) and our learning-based algorithm (BRTDP). The next two columns compare the number of states that BRTDP explored (\#States\_B) to the total number of states in the model. The rightmost column shows the number of MSECs in the model.}\vspace{-0.5em} \label{table:time}
	\medskip
	
	 \makebox[\linewidth]{ 
\begin{tabular}{|c|c||c|c|c||c|c|c|} \hline 
Model & Parameters  & PRISM& BVI & BRTDP & \#States\_B & \#States & \#MSECs \\ \hline \hline
\multirow{2}{*}{mdsm} 
& prop=1 
& 8 
& 8 
 & 17
 & 767 & 62,245
&1
\\ \cline{2-8}
& prop=2 
& 4 
& 4 
 & 29
 & 407 & 62,245
&1
\\ \cline{2-8}
\hline
\multirow{1}{*}{cdmsn} 
&
& 2 
& 2 
 & 3
 & 1,212 & 1,240
& 1
\\ \cline{2-8}
\hline
\multirow{2}{*}{cloud} 
& N=5 
& 3 
& 7 
 & 15
 & 1,302 & 8,842
& 4,421
\\ \cline{2-8}
& N=6 
& 6 
& 59 
 & 4
 & 570 & 34,954
& 17,477
\\ \cline{2-8}
\hline
\end{tabular} } \vspace*{-1em}\end{table}

%% file: conclusion.tex
\section{Conclusions} \label{sec:concl}

\vspace{-0.7em}
We have provided the first stopping criterion for value iteration on simple stochastic games and an anytime algorithm with bounds on the current error (guarantees on the precision of the result).
The main technical challenge was that states in end components in SG can have different values, in contrast to the case of MDP.
We have shown that collapsing is in general not possible, but we utilized the analysis to obtain the procedure of \emph{deflating}, a solution on the original graph.
Besides, whenever a SEC is identified for sure it can be collapsed and the two techniques of collapsing and deflating can thus be combined.

The experiments indicate that the price to pay for the overhead to compute the error bound is often negligible.
For each of the available models, at least one of our two implementations has performed similar to or better than the standard approach that yields no guarantees.
Further, the obtained guarantees open the door to (e.g. learning-based) heuristics which treat only a part of the state space and can thus potentially lead to huge improvements.
Surprisingly, already our straightforward adaptation of such an algorithm for MDP to SG yields interesting results, palliating the overhead of our non-learning method, despite the most naive implementation of deflating.
Future work could reveal whether other heuristics or more efficient implementation can lead to huge savings as in the case of MDP \cite{atva}.
 
%
%
%

%% file: app.tex

\section{Technical information and pseudocode}

\subsection{Definition of $\COLLAPSEFUN$}\label{app:coll}
\input{collapsing}

\subsection{Pseudocode for Bellman update}\label{sec:UPDATEVI}
\input{update}

\subsection{Example of a converging EC}\label{app:convEC}
Figure \ref{ex:convEC} depicts an EC of an SG where $\gub=\val$ and thus BVI converges for this EC.
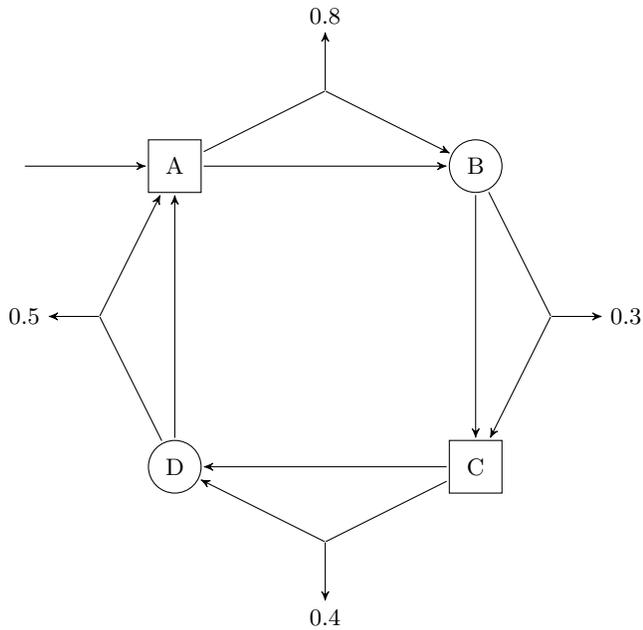
\begin{figure}
\input{convergingEC}
\caption{An example of an EC where $\gub=\val$. 
The numbers on the outgoing edges represent the values of the successor states. In this EC BVI converges, since it is no $\CEC$ and there is only a single solution for the Bellman equations.
Using the leaving action from $A$, the maximizing player can achieve a value of 0.8 (with a probability of one half), which is the best possible value. The minimizer can achieve the least possible value of 0.3 (with probability one half) from $B$. If either of the players decided not to leave, his opponent would get infinitely many chances at leaving, and hence a suboptimal value would be achieved. Thus both players pick their leaving actions and BVI converges for this EC.}
\label{ex:convEC}
\end{figure}

\subsection{Additional information to the BRTDP algorithm}\label{app:brtdp}
\input{brtdp}

\section{Experimental setup}
\input{exp}

\newpage
\section{Proofs}

\input{illusion}

\input{convBVIwo}

\input{find}

\input{corrADJ}

\input{BVIA_term}




%% file: collapsing.tex
$\COLLAPSEFUN$ is not only able to collapse ECs in MDPs, but also SECs in SGs. Note that every EC in an MDP is a SEC. 
If there are no actions of Maximizer leaving the SEC, since the game is non-blocking, we have to keep staying actions, so the SEC becomes a sink.
Otherwise all staying actions are removed and the SEC becomes a single state with all leaving actions of states in the SEC available.

\begin{definition}[$\COLLAPSEFUN$]
Let $\exGame$ be an $\SG$ and $T$ a SEC in $\G$. Then 
$\COLLAPSEFUN(\G,T) = \exGame'$, where $\G'$ is defined as follows:\\
\begin{itemize}
\item $\states' = (\states \setminus T) \cup \{\state<T>\}$
\item $\states<\Box>' = (\states<\Box> \setminus T) \cup \set{\state<T>}$
\item $\states<\circ>' = (\states<\circ> \setminus T)$
\item $\initstate' = \begin{cases} \state<T> &\mbox{if } (\initstate \in T) \\ \initstate &\mbox{else}\end{cases}$
\item $\actions' = \actions$
\item $\Av'(\state)$ is defined for all $\state \in \states'$ by:
	\begin{itemize} 
	\item $\Av(\state)$,\\
	 if $\state \in (\states \setminus T)$, i.e. $\state \neq \state<T>$\\ (Rest stays the same)
	\item $\bigcup_{t \in T} \set{\action \mid \action \in \Av(t) \wedge (t,\action) \stays T}$ \\
	 if  $\state=\state<T> \wedge \bigcup_{t \in T_\Box} \set{\action \mid \action \in \Av(t) \wedge (t,\action) \leaves T} = \emptyset$ \\
	 (Keep staying actions, if there is no exit for Maximizer)
	\item $\bigcup_{t \in X_\Box} \set{\action \mid \action \in \Av(t) \wedge (t,\action) \leaves X}$,\\
	if $\state=\state<T>$ \\
	(Keep leaving actions, if there is an exit for Maximizer)
	\end{itemize}
\item $\trans'$ is defined for all $\state \in \states'$ and $\action \in \Av'(\state)$ by:
	\begin{itemize}
	\item $\trans'(\state,\action)(\state') = \trans(\state,\action)(\state')$\\
	for all $\state' \in \states'$ with $\state,\state' \neq \state<T>$ \\
	(Rest stays the same)
	\item $\trans'(\state,\action)(\state<T>) = \sum_{\state' \in X} \trans(\state,\action)(\state')$,\\
	if $\state \neq \state<T>$ \\
	(going to $X$)
	\item $\trans'(\state<T>,\action)(\state') = \trans(\state,\action)(\state')$\\
	 for all $\state' \in \states' \setminus \set{\state<T>}$ \\
	 (leaving from $X$)
	\item $\trans'(\state<T>,\action)(\state<T>) =\sum_{t \in T} \trans(\state,\action)(t)$\\
	 for all $\state \in T$ such that  $\action \in \Av(\state)$ \\
	 (staying in $T$)
	\end{itemize}
\end{itemize}
\end{definition}

%% file: update.tex
This is the pseudocode for the usual update procedure used for value iteration. It amounts to performing Bellman updates, i.e. applying Equations \ref{eq:Lsa} and \ref{eq:Ls} to the functions $\ub$ and $\lb$ once.
\begin{algorithm}[H]
\caption{$\UPDATE$ procedure for (naive) value iteration on SG}\label{alg:UPDATEVI}
\begin{algorithmic}[1]
\Procedure{$\UPDATE$}{$\lb:\states\to[0,1]$, $\ub:\states\to[0,1]$}
\For {$\state \in \states<\Box>$}
	\State $\lb(\state) \gets 
			\max_{\action \in \Av(\state)} \sum_{s' \in S} \trans(\state,\action,\state') \cdot \lb(\state')$
	\State $\ub(\state)\gets\max_{\action \in \Av(\state)} \sum_{s' \in S} \trans(\state,\action,\state') \cdot \ub(\state') $
\EndFor
\smallskip
\For {$\state \in \states<\circ>$}
	\State 
	$\lb(\state) \gets 
		\min_{\action \in \Av(\state)} \sum_{s' \in S} \trans(\state,\action,\state') \cdot \lb(\state') $
	\State
	$\ub(\state) \gets 
			\min_{\action \in \Av(\state)} \sum_{s' \in S} \trans(\state,\action,\state') \cdot \ub(\state') $
\EndFor
\EndProcedure
\end{algorithmic}
\end{algorithm}

%% file: convergingEC.tex
\begin{tikzpicture}
\drawdummy (init) at (-2,4) {};

\drawbox (A) at (0,4) {A};
\drawcirc (B) at (4,4) {B};
\drawbox (C) at (4,0) {C};
\drawcirc (D) at (0,0) {D};

\drawdummy (AB) at (2,5) {};
\drawdummy (BC) at (5,2) {};
\drawdummy (CD) at (2,-1) {};
\drawdummy (DA) at (-1,2) {};

\node (ABo) at (2,6) {0.8};
\node (BCo) at (6,2) {0.3};
\node (CDo) at (2,-2) {0.4};
\node (DAo) at (-2,2) {0.5};

\draw[->] (init) to (A);

\draw[->] (A) to (B);
\draw[->] (B) to (C);
\draw[->] (C) to (D);
\draw[->] (D) to (A);

\draw[-] (A) to (AB);
\draw[-] (B) to (BC);
\draw[-] (C) to (CD);
\draw[-] (D) to (DA);

\draw[->] (AB) to (ABo);
\draw[->] (AB) to (B);
\draw[->] (BC) to (BCo);
\draw[->] (BC) to (C);
\draw[->] (CD) to (CDo);
\draw[->] (CD) to (D);
\draw[->] (DA) to (DAo);
\draw[->] (DA) to (A);

\end{tikzpicture}

%% file: brtdp.tex
This section provides a more detailed pseudocode and some intuitions for Algorithm \ref{alg:BRTDP}. It is based on~\cite[Algorithm 7]{MatPhD}. 

We introduce the notion of $\best$ actions to avoid having a case distinction over the players. 
\begin{definition}[$\best$ actions]\label{def:best}
The set of $\best$ actions for a state $\state$, given the the current $\ub$ and $\lb$ \\
$\best_{\ub,\lb}(\state) := 
\begin{cases} \set{\action \in \Av(\state) \mid \ub(\state,\action) = \max_{\action \in \Av(\state)}\ub(\state,\action)}		&\mbox{if } \state \in \states<\Box>  \\
\set{\action \in \Av(\state) \mid \lb(\state,\action) =\min_{\action \in \Av(\state)}\lb(\state,\action)} &\mbox{if } \state \in \states<\circ>. \end{cases}$,
\end{definition}

We also need the notion of a restricted game, because we want to restrict computation to only the explored state space.
\begin{definition}[Restricted game $\mathsf{G}^{\Vis}$]
Let $\exGame$ be an $\SG$ and $\Vis \subseteq \states$.
Then $\exGame[\Vis]$, i.e. the game restricted to $\Vis$, is defined as follows:
 \begin{itemize}
\item $\states<\Box>[\Vis] = \Vis_\Box \cup \set{t \mid \exists \state \in \Vis, \action \in \Av(\state)$ such that $\trans(\state,\action,t)>0 \wedge t \in \states<\Box>}$
\item $\states<\circ>[\Vis] = \Vis_\circ \cup \set{t \mid \exists \state \in \Vis, \action \in \Av(\state)$ such that $\trans(\state,\action,t)>0 \wedge t \in \states<\circ>}$
\item $\states[\Vis] = \states<\Box> \cup \states<\circ>$
\item $\actions[\Vis] = \bigcup_{\state \in \Vis} \Av(\state) \cup \set{\bot}$
\item $\Av[\Vis](\state) = \Av(\state)$ if $\state \in \Vis$, else $\Av[\Vis](\state) = \bot$
\item $\trans[\Vis](\state,\action) = \trans(\state,\action)$ if $\state \in \Vis$, else $\trans(\state,\bot)(\state) = 1$
\end{itemize}
\end{definition}

The following pseudocode is for the complete BRTDP version of BVI on SGs, not only the update procedure. This is because we have to initialize the additional variable $\Vis$ to remember the set of visited states. After the initialization, it runs sample trials until convergence.

\begin{algorithm}[H]
\caption{Complete BRTDP version of BVI on SGs}\label{alg:BRTDPapp}
\begin{algorithmic}[1]
\Procedure{BRTDP}{$\SG$ $\G$, precision $\epsilon$}
\State \Comment{Initialization}
\State $\ub(\state) = 1$ for all $\state \in \states$
\State $\lb(\state) = 0$ for all $\state \in \states$
\State $\lb(\target) = 1$
\State $\Vis \gets \emptyset$ \Comment{Set of states visited so far.}
\smallskip
\Repeat 
	\State $\rSTA(\initstate,\ub,\lb,\Vis,\epsilon)$
	\State $i \gets i+1$
\Until{$\ub(\initstate) - \lb(\initstate) < \epsilon$}\label{line:BRTDPterm}
\EndProcedure
\end{algorithmic}
\end{algorithm}

A sample trial is divided into three phases: Simulating, updating and deflating.

In the simulation phase, starting from the initial state a path through the game gets sampled, always picking a $\best$ action and then sampling a successor of that action according to the $\GETSUCC$-function. This function is not explicitly stated here, but different versions are discussed in~\cite{atva}. For example, one can sample the successor according to the probability distribution or one can pick the successor with the maximal difference between its bounds in order to get to regions of the state space that have not yet been done. 
The set of visited states is saved in $\Vis$, in order to be able to compute the restricted game later in the deflating phase. 
To avoid the simulation getting stuck in an EC, there is the condition in Line~\ref{line:explorebreak}. For the function $\EXPLOREBREAK$ several versions are possible, too. One can break the simulation as soon as one state gets seen the second time on a path or after a certain number of steps $k$. $k$ should be larger than the current explored state space size in order to ensure that new states can be visited with positive probability. The formulas for $k$ given in~\cite{MatPhD} result in numbers much larger than the current explored state space, so that the paths get very long and if the simulation is stuck a lot of time is wasted. Our experiments have shown that the simple formula $2 \cdot |\states[\Vis]|$ is very fast.

In the update phase, the bounds for all state-action-pairs on the path get updated according to the Bellman equations. To write this concisely, we utilize the notation of $\best$ actions.

The deflating phase works exactly as in Algorithms \ref{alg:BVIA} and \ref{alg:BRTDP}, only that now we apply $\FIND$ to the restricted game. We have noted this in the pseudocode by adding the parameter $\G[\Vis]$, thereby specifying the implicit parameter of the SG that $\FIND$ is applied to.

Note that we have two functions that we use for accessing the elements of the path $\path$, namely $last$ and $pop$. $last$ returns the last element of the path without removing it, $pop$ returns that element and removes it from $\path$.

Many optimizations can be applied to the algorithm, for example stopping the simulation not only in $\target$ and $\sink$, but in any state where the bounds have converged. One can also only deflate if the simulation was stopped due to $\EXPLOREBREAK$, because otherwise the simulation was not stuck and quite possibly deflating is not needed. Computing the MEC decomposition of $\G[\Vis]$ before finding the MSECs is also speeding up the computation. 
Some other optimizations are only helpful in certain cases, for example only adjusting if the simulation has been broken several times or only adjusting the last EC on the path.
\newpage

\begin{algorithm}[H]
\caption{Procedure for sampling and updating}\label{alg:rst}
\begin{algorithmic}[1]
\Procedure{$\rSTA$}{state $\state$, upper bound function $\ub$, lower bound function $\lb$, set of states $\Vis$, precision $\epsilon$}
\State \Comment{Simulation phase}
\State $\path \gets \initstate$
\Repeat
	\State $\action \gets $ sampled from $\best_{\ub,\lb}(last(\path))$
	\State $\state \gets \GETSUCC (last(\path),\action)$
	\State $\path \gets \path \: \action \: \state$
	\State $\Vis \gets \Vis \cup \state$
	\If {$\EXPLOREBREAK(\path)$} \label{line:explorebreak}
		\State \textbf{break}
	\EndIf
\Until{$\state \in \set{\target,\sink}$}
\smallskip
\smallskip
\State \Comment{Update phase}
\State $pop(\path)$ \Comment{Remove the last state} \label{BRTDP_Apop}
\Repeat
	\State $\state \gets pop(\path)$
	\State $\action \gets pop(\path)$
	\State $\ub(\state,\action) \gets \sum_{\state' \in \states} \trans(\state,\action) (\state') \cdot \ub(\state')$
	\State $\lb(\state,\action) \gets \sum_{\state' \in \states} \trans(\state,\action) (\state') \cdot \lb(\state')$
\Until {$\path$ is empty}
\smallskip
\smallskip
\State \Comment{Deflating phase}
\State compute $\G[\Vis]$\label{line:gVis}
\For{$T \in \FIND(\lb,\G[\Vis])$} \label{line:loopAdjust2}
	\State $\ADJUST(T,\ub)$ \label{line:callAdjust2}
\EndFor
\EndProcedure
\end{algorithmic}
\end{algorithm}

%% file: exp.tex
\subsection{Technical details}\label{app:tech}
The experiments were conducted on a server with 256 GB RAM and 2 Intel(R) Xeon(R) E5-2630 v4 2.20 GHz processors. However, computation was limited to one core to avoid results being incomparable due to different times spent parallelizing. 
All model checkers worked at a precision of $\epsilon=10^{-6}$. Each experiment had a timeout of 15 minutes. An X in a table indicates that the model checker was unable to finish the computation in the time limit. 
We set the available Java memory to 16GB to enable the solvers to construct the models, although still the largest versions of csma and mer could not be loaded with the explicit engine.
Since the simulation based approaches are randomized, we took the median of 20 repetitions of the experiments.
For SGs, we ran the experiments both with the hybrid and the explicit engine of PRISM-games. However, since the models are small, i.e. less than 100,000 states, verification times for the hybrid and explicit approach of PRISM never differ by more than a second, and hence we only included the hybrid engine in Table \ref{table:time}.

\subsection{Models}\label{sec:models}
Our experiments are based on the ones that were conducted in \cite{MatPhD}. Most models we use are also analyzed in that thesis, and we obtained them from the website~\cite{modelsource}, where the author of the thesis made them available for download. We used the exact models from that website, but partly modified the properties to be of a form that our implementation can handle. These modifications did not change the semantics of the property, e.g. instead of formulating a property that a probability is greater than a certain number(P$>$0.999) we compute the maximal probability (Pmax=?), and then manually check whether it is greater than the number. 
The only models not from \cite{MatPhD} are csma and leader, which are part of the examples included in PRISM 4.4.

We consider six MDP models, namely firewire, wlan, zeroconf, csma, leader and mer. The first five are part of the PRISM benchmark suite \cite{PRISMben}, mer is from~\cite{mer}.  The four $\SG$ models are mdsm, cdmsn, team-form and cloud, and the first three are contained in the PRISM-games case studies~\cite{gamesBen}. Cloud is from~\cite{cloud}. Note that some of the $\SG$ models actually contain more than two players. However, since there are at most two coalitions, they can be viewed as an $\SG$ with only two players. We will now shortly describe all models, the properties we check and the parameters we use for scaling.

\subsubsection{firewire}~\cite{firewire}:\\
This case study models the protocol known as FireWire, which is a leader election protocol of the IEEE 1394 High Performance Serial Bus. Several devices connected to a bus can use the protocol to dynamically elect a leader.
We compute the probability Pmax=?~[F\nolinebreak~leader\_ {\nolinebreak}elected], so the maximal probability with which a leader gets elected before the deadline. By this one can check the property that a leader gets elected with a certain, optimally high, probability. To scale the model up, we raise the deadline.  

\subsubsection{wlan}~\cite{wlan}:\\
This model describes the two-way handshake mechanism of the IEEE 802.11 medium access control (WLAN protocol). Two stations try to communicate with each other; however, if both of them send at once, a collision occurs. 
We are interested in computing the maximum probability that both stations transmit their messages correctly, i.e. Pmax=?~[F~s1=12~\&~s2=12], where s1 and s2 describe the state of the stations, and 12 is the final state where the transmission was successful. To scale the model up, we increase the maximal backoff $k$ and the maximal number of collisions COL.

\subsubsection{zeroconf}~\cite{zeroconf}:\\
Zeroconf is a protocol for dynamically assigning an IP address to a device, provided that several other hosts have already blocked some IP addresses. The device picks some IP randomly and then sends probes to check whether this address is already in use. The parameters that we use to scale the model are $N$, the number of other hosts already possessing an IP address and $K$, the number of probes sent. The probability we are interested in is Pmin=?~[F~configured], so the minimum probability with which the device obtains an IP address. 

\subsubsection{csma}:~\cite{PRISMben}\\
This case study concerns the IEEE 802.3 CSMA/CD (Carrier Sense, Multiple Access with Collision Detection) protocol. N is the number of stations and K is the maximum backoff. Pmin=? [ F min\_backoff\_after\_success$<$=K ] is the probability we are interested in, namely that a message of some station is eventually delivered before k backoffs.

\subsubsection{leader}~\cite{PRISMben}:\\
 This case study is based on the asynchronous leader election protocol of \cite{leader}. This protocol solves the following problem.
 Given an asynchronous ring of N processors design a protocol such that they will be able to elect a leader (a uniquely designated processor) by sending messages around the ring.  The probability we are interested in is Pmax=? [ F "elected" ], so that at some point a leader is elected.

\subsubsection{mer}~\cite{mer}:\\
In the Mars Exploration Rover there is a resource arbiter that handles distributing resources to different users. There is a probability $x$ that the communication between the arbiter and the users fails. We change this probability to influence the structure of the MDP. The probability we compute is Pmax=? [F~err\_G], which is the maximum probability that an error occurs. 

\subsubsection{mdsm}~\cite{mdsm}:\\
This case study models multiple households which all consume different amounts of energy over time. To minimize the peak energy consumption, they utilize the distributed energy management ``Microgrid Demand-Side Management'' (mdsm). The property we check is the maximal probability with which the first household can deviate from the management algorithm, i.e. Pmax=? [F~deviated], which should be smaller than 0.01. We check the property once for player 1 and once for player 2.  

\subsubsection{cdmsn}~\cite{mdsm,cdmsn}:\\
This model describes a set of agents which have different sites available and different preferences over these sites. The collective decision making algorithm of this case study is utilized so that the agents agree on one decision. We analyze the model to find the strategy for player 1 to make the agents agree on the first site with a high probability, so $<<$p1$>>$~Pmax=? [F all\_prefer\_1]. 

\subsubsection{team-form}~\cite{teamform}:\\
As in the previous case study, there are a set of agents in a distributed environment. They need to form teams so they are able to perform a set of tasks together. We want to compute a strategy so that the first task is completed with the maximal possible probability, so we check the property $<<$p1,p2,p3$>>$ Pmax=? [F task1\_completed]. We scale the model using the number of agents N. 

\subsubsection{cloud}~\cite{cloud}:\\
This model describes several servers and virtual machines forming a cloud system. The controller of the system wants to deploy a web application on one of the virtual machines, but it is possible that the servers fail due to hardware failures. We compute the strategy and the maximal probability for the controller to successfully deploy his software, so $<<$controller$>>$ Pmax=? [F~deployed]. The model can be scaled by increasing the number of virtual machines N.\\\\

\subsection{Implementation optimizations}
We tried several optimizations for both our algorithms, e.g. deflating an MSEC repeatedly until no change greater than the precision occurs instead of changing the upper bounds only once, only deflating the last EC that BRTDP explored and prefering exiting actions in BRTDP. None of these showed a significant improvement.

However, one optimization may greatly influence the verification time, namely changing how often we execute the deflating step.
For correctness, it is only required that deflation is called regularly; by that we mean infinitely often, if the algorithm would not terminate at some point. In different words: For each iteration it holds, that in a future iteration deflate is executed, except for the iterations between between the last deflating step and termination. For those it only holds, that a deflating step would have occurred, had the algorithm not terminated.
We executed $\ADJUST$ only every $n$ steps for different $n$ between 1 and 1000. 
Depending on the model, the choice of $n$ may be irrelevant or influence performance even by an order of magnitude. Furthermore, both a large $n$ (100) or a small $n$ (1) can be the best choice.
It seems sensible to not execute the expensive deflation after every Bellman update, when there are not many MSECs and information just needs to be propagated. However, if many MSECs are ordered sequentially, steps without a deflation are mostly useless. This is why information on the model is needed to decide how often $\ADJUST$ should be called. Maybe it is possible to decide this during the value iteration, but we have not yet found any heuristic for that.

\subsection{Experiments on MDPs}\label{sec:MDPexp}
For the experiments on the MDPs we used four different programs, namely PRISM 4.4.beta, \textsc{Storm} 1.1.0~\cite{Storm}, the implementation that was used in \cite{atva} (called BRTDP\_coll) and our own implementation, i.e. BVI and BRTDP. We ran PRISM in three different version, namely once with the hybrid engine (called PRISM\_h), once with the explicit engine (called PRISM\_e) and once with interval iteration (called PRISM\_ii), so the approach using collapsing to give guarantees, with the explicit engine. The first configuration shows what PRISM can achieve, the second and third are fair competitors for BVI, since all of them have to construct the whole model explicitly.

Tables \ref{table:timeMDP} and \ref{table:statesMDP} show the results of the experiments, namely the complete verification times for all models and the state space size with the number of states that the simulation based approaches explored.\\

The results in Table \ref{table:timeMDP} show that collapsing and deflating perform quite similar, since PRISM\_ii and BVI as well as BRTDP\_coll and BRTDP produced similar verification times. 

On smaller models, e.g. wlan for the first four rows, csma the first four rows and leader the first three rows, PRISM\_ii and BVI are not significantly slower than PRISM\_h. For the other models, the gain of the hybrid engine makes PRISM\_d and \textsc{Storm} a lot faster. 
Most probably when implementing BVI for the hybrid engine the overhead for giving the guarantee will also be small. 
In all models but zeroconf for K=10 and mer for x=0.0001, PRISM\_e and BVI produce times in the same order of magnitude, so the overhead for computing the guarantees is not too large.

The simulation based approaches BRTDP and BRTDP\_coll perform well on firewire, wlan and zeroconf, even outperforming \textsc{Storm} in some cases. For firewire they are two orders of magnitude faster. So in certain cases, simulation based approaches can produce a huge speedup while still giving guarantees.
However for csma, leader and mer they are not well suited, as they need to explore thousands of states to achieve convergence. For the first two rows of mer they are still faster, since the explored part of the state space is very small in comparison to the whole model and not too large in general, but as the model is scaled, the number of relevant states grows too large for the simulation based approaches to work well.

The results in Table \ref{table:statesMDP} show that BRTDP\_coll explores a larger portion of the state space for almost all experiments. This is due to the different choice of $k$, so the number of steps before an exploration is broken. The higher $k$ that is implemented in BRTDP\_coll allows for exploring longer, and hence more of the state space is explored. This can be advantageous for the verification time, as in mer with x=0.0001 or for csma with N=3 K=4, but can also prevent from producing any result in time as in leader with N=6 or mer with x=0.1. 

In general one can see, that the approach of deflating works well also on MDPs and that giving guarantees is often possible without significant overhead.

\input{timeMDP}
\input{statesMDP}

\subsection{Experiments on SGs}\label{sec:SGexp}
The results of the experiments on teamform, where pre-computation already solved the problem for PRISM-games and BVI, are shown in Table \ref{table:teamform}.
BRTDP performs bad on this model because a large part of the state space is relevant for convergence, namely over 60\% with N=3 and at least 50\% for N=4. Also the number of relevant states is very large.

\begin{table}[H] \centering \small \caption{Verification time for teamform in seconds. An X denotes that the computation did not finish within the time limit. } \label{table:teamform} \makebox[\linewidth]{ \begin{tabular}{|c|c||c|c|c|} \hline Model & Parameters &  PRISM& BVI & BRTDP \\ \hline \hline
\multirow{2}{*}{teamform} 
& N=3 
& 3 
& 3 
 & 139
\\ \cline{2-5}
& N=4 
& 9 
& 10
& X 
\\ \cline{2-5}
\hline
\end{tabular} } \end{table}

%% file: timeMDP.tex
\begin{table}[] \centering \small \caption{CPU time for each experiment in seconds. Models and their scaling parameters are denoted on the left, solvers in the topmost column. We compared Storm, PRISM in three different version, namely once with the hybrid engine (called PRISM\_h), once with the explicit engine (PRISM\_e) and once with interval iteration (PRISM\_ii), furthermore our approaches (BVI and BRTDP) and the collapsing based approach from \cite{atva} (BRTDP\_coll)} \label{table:timeMDP} \makebox[\linewidth]{ \begin{tabular}{|c|c||c|c||c|c|c||c|c|} \hline Model & Parameters & Storm & PRISM\_h & PRISM\_e & PRISM\_ii & BVI & BRTDP\_coll & BRTDP \\ \hline \hline
\multirow{4}{*}{firewire} 
& deadline=220 delay=36
& 162 
& 259 
& 459 
& 449 
& 468 
 & 2
 & 2
\\ \cline{2-9}
& deadline=240 delay=36
& 219 
& 453 
& 600 
& 593 
& 718 
 & 2
 & 2
\\ \cline{2-9}
& deadline=260 delay=36
& 252 
& 882 
& 745 
& X 
& X 
 & 2
 & 2
\\ \cline{2-9}
& deadline=280 delay=36
& 316 
& 751 
& X 
& X 
& X 
 & 2
 & 2
\\ \cline{2-9}
\hline
\multirow{6}{*}{wlan} 
& k=2 COL=2
& 0 
& 3 
& 3 
& 4 
& 3 
 & 2
 & 2
\\ \cline{2-9}
& k=2 COL=6
& 1 
& 4 
& 5 
& 7 
& 6 
 & 2
 & 2
\\ \cline{2-9}
& k=4 COL=2
& 4 
& 16 
& 18 
& 25 
& 22 
 & 2
 & 2
\\ \cline{2-9}
& k=4 COL=6
& 7 
& 21 
& 34 
& 35 
& 35 
 & 2
 & 2
\\ \cline{2-9}
& k=6 COL=2
& 57 
& 164 
& 703 
& 737 
& 727 
 & 2
 & 2
\\ \cline{2-9}
& k=6 COL=6
& 58 
& 179 
& 700 
& 740 
& 754 
 & 2
 & 2
\\ \cline{2-9}
\hline
\multirow{6}{*}{zeroconf} 
& K=2 N=20
& 1 
& 6 
& 5 
& 5 
& 16 
 & 3
 & 2
\\ \cline{2-9}
& K=2 N=500
& 0 
& 7 
& 5 
& 6 
& 21 
 & 3
 & 2
\\ \cline{2-9}
& K=2 N=1000
& 1 
& 7 
& 5 
& 5 
& 24 
 & 3
 & 2
\\ \cline{2-9}
& K=10 N=20
& 35 
& 167 
& 101 
& 128 
& X 
 & 3
 & 2
\\ \cline{2-9}
& K=10 N=500
& 36 
& 172 
& 101 
& 138 
& X 
 & 4
 & 2
\\ \cline{2-9}
& K=10 N=1000
& 35 
& 180 
& 100 
& 145 
& X 
 & 5
 & 2
\\ \cline{2-9}
\hline
\multirow{6}{*}{csma} 
& N=2 K=2
& 0 
& 1 
& 2 
& 2 
& 2 
 & 3
 & 2
\\ \cline{2-9}
& N=2 K=4
& 0 
& 2 
& 2 
& 2 
& 2 
 & 7
 & 7
\\ \cline{2-9}
& N=2 K=6
& 0 
& 2 
& 4 
& 4 
& 4 
 & 89
 & 86
\\ \cline{2-9}
& N=3 K=2
& 0 
& 2 
& 4 
& 4 
& 4 
 & 14
 & 22
\\ \cline{2-9}
& N=3 K=4
& 15 
& 5 
& 31 
& 32 
& 37 
 & 227
& X 
\\ \cline{2-9}
& N=3 K=6
& X 
& 58 
& X 
& X 
& X 
& X 
& X 
\\ \cline{2-9}
\hline
\multirow{4}{*}{leader} 
& N=3 
& 0 
& 1 
& 1 
& 1 
& 1 
 & 2
 & 2
\\ \cline{2-9}
& N=4 
& 0 
& 2 
& 2 
& 2 
& 2 
 & 3
 & 4
\\ \cline{2-9}
& N=5 
& 0 
& 3 
& 3 
& 3 
& 3 
 & 24
 & 13
\\ \cline{2-9}
& N=6 
& 4 
& 10 
& 9 
& 10 
& 9 
& X 
 & 46
\\ \cline{2-9}
\hline
\multirow{6}{*}{mer} 
& x=0.0001 n=1500
& 48 
& 108 
& 184 
& X 
& X 
 & 17
 & 41
\\ \cline{2-9}
& x=0.0001 n=3000
& 101 
& 220 
& X 
& X 
& X 
 & 16
 & 40
\\ \cline{2-9}
& x=0.1 n=1500
& 51 
& 150 
& 188 
& X 
& X 
& X 
 & 698
\\ \cline{2-9}
& x=0.1 n=3000
& 102 
& 291 
& X 
& X 
& X 
& X 
 & 604
\\ \cline{2-9}
\hline
\end{tabular} } \end{table}

%% file: statesMDP.tex
\begin{table}[] \centering \small \caption{The number of states for each model and the number of states that the simulation based approach had to explore. Models and their scaling parameters are denoted on the left. (top is state space and simulation based approach)} \label{table:statesMDP} \makebox[\linewidth]{ \begin{tabular}{|c|c||c||c|c|} \hline Model & Parameters & \#States & BRTDP\_coll & BRTDP \\ \hline \hline
\multirow{4}{*}{firewire} 
& deadline=220 delay=36
& 10,490,495 
 & 792
 & 797
\\ \cline{2-5}
& deadline=240 delay=36
& 13,366,666 
 & 779
 & 718
\\ \cline{2-5}
& deadline=260 delay=36
& 15,255,584 
 & 791
 & 432
\\ \cline{2-5}
& deadline=280 delay=36
& 19,213,802 
 & 1,050
 & 683
\\ \cline{2-5}
\hline
\multirow{6}{*}{wlan} 
& k=2 COL=2
& 28,598 
 & 584
 & 113
\\ \cline{2-5}
& k=2 COL=6
& 107,854 
 & 676
 & 110
\\ \cline{2-5}
& k=4 COL=2
& 345,118 
 & 767
 & 112
\\ \cline{2-5}
& k=4 COL=6
& 728,990 
 & 764
 & 120
\\ \cline{2-5}
& k=6 COL=2
& 5,007,666 
 & 858
 & 125
\\ \cline{2-5}
& k=6 COL=6
& 5,007,670 
 & 691
 & 113
\\ \cline{2-5}
\hline
\multirow{6}{*}{zeroconf} 
& K=2 N=20
& 89,586 
 & 393
 & 125
\\ \cline{2-5}
& K=2 N=500
& 89,586 
 & 1,601
 & 560
\\ \cline{2-5}
& K=2 N=1000
& 89,586 
 & 1,625
 & 937
\\ \cline{2-5}
& K=10 N=20
& 3,001,911 
 & 1,161
 & 645
\\ \cline{2-5}
& K=10 N=500
& 3,001,911 
 & 3,836
 & 643
\\ \cline{2-5}
& K=10 N=1000
& 3,001,911 
 & 5,358
 & 622
\\ \cline{2-5}
\hline
\multirow{6}{*}{csma} 
& N=2 K=2
& 1,038 
 & 964
 & 966
\\ \cline{2-5}
& N=2 K=4
& 7,958 
 & 7,691
 & 7,811
\\ \cline{2-5}
& N=2 K=6
& 66,718 
 & 64,341
 & 32,929
\\ \cline{2-5}
& N=3 K=2
& 36,850 
 & 22,883
 & 26,632
\\ \cline{2-5}
& N=3 K=4
& 1,460,287 
 & 266,724
 &X
\\ \cline{2-5}
& N=3 K=6
& 84,856,004 
 & X
& X 
\\ \cline{2-5}
\hline
\multirow{4}{*}{leader} 
& N=3 
& 364 
 & 335
 & 306
\\ \cline{2-5}
& N=4 
& 3,172 
 & 2,789
 & 2,565
\\ \cline{2-5}
& N=5 
& 27,299 
 & 21,550
 & 8,420
\\ \cline{2-5}
& N=6 
& 237,656 
 & 128,593
 & 20,551
\\ \cline{2-5}
\hline
\multirow{6}{*}{mer} 
& x=0.0001 n=1500
& 8,862,064 
 & 2,603
 & 2,005
\\ \cline{2-5}
& x=0.0001 n=3000
& 17,722,564 
 & 2,632
 & 2,035
\\ \cline{2-5}
& x=0.1 n=1500
& 8,862,064 
 & X
 & 7,461
\\ \cline{2-5}
& x=0.1 n=3000
& 17,722,564 
 & X
 & 7,453
\\ \cline{2-5}
\hline
\end{tabular} } \end{table}

%% file: illusion.tex
\subsection{Proof of Lemma \ref{lem:illu}}

\setcounter{lemma}{0}

\begin{lemma}\label{lem:illu} 
Let $T$ be an EC of an \SG.
Then for every $m$ such that
$\exit<\val>[\Box](T) \leq m \leq \exit<\val>[\circ](T)$,
there is a solution to the Bellman equations, which on $T$ is constant and equal to $m$.
\end{lemma}

\begin{proof}
Let $T$ be an EC and $m$ a number in $[0, 1]$ satisfying $\exit<\val>[\Box](T) \leq m \leq \exit<\val>[\circ](T)$. 
Furthermore, let $\ub$ be a function where for all $\state \in T: \ub(\state)=m$. 

We now show that $\ub$ is a fixpoint of the Bellman equations, i.e. for each 
$\state \in T_\Box: \ub(\state) =  \max_{\action \in \Av(\state)}\ub(\state,\action) = m$ and for each $\state \in T_\circ: \ub(\state) =  \min_{\action \in \Av(\state)}\ub(\state,\action) = m$.
\vspace{2em}

\noindent \textbf{Claim 1:} For every state $\state \in T$ there is an action $\action \in \Av(\state)$ such that $\ub(\state,\action)=m$.
\vspace{-1em}
\begin{enumerate}
\item Let $\state \in T$. 
\item Since $T$ is an EC, there exists an $\action \in \Av(\state): (\state,\action) \stays T$. Let $\action$ be such a staying action.
\item \label{val_a} From this we know $\post(\state,\action) \subseteq T$. Since by assumption for all $\state' \in T: \ub(\state')=m$, it also holds that
	\[ \ub(\state,\action) = \sum_{\state' \in \post(\state,\action)} \trans(\state,\action)(\state') \cdot \ub(\state') =\sum_{\state' \in \post(\state,\action)} \trans(\state,\action)(\state') \cdot m = m\]
\end{enumerate}

It remains to show that $\action$ actually is the action that is used to compute the upper bound. For that we make a case distinction on the player that $\state$ belongs to.
\begin{enumerate}
	\item $\state \in T_\Box$: \\
	Since $\action$ is an arbitrary staying action, it suffices to show that for all actions $\action' \in \Av(\state): (\state,\action') \leaves\ T \implies \ub(\state,\action) \geq \ub(\state,\action')$. This is done by the following chain of equations:
		\begin{align*}
		\ub(\state,\action) &= m \tag{by Claim 1}\\
			&\geq \exit<\ub>[\Box](T) \tag{by assumption} \\
			&\geq \max_{\action' \in \Av(\state): (\state,\action') \leaves T} \ub(\state,\action') \tag{since $\state \in T_\Box$} \\
			&\geq  \ub(\state,\action') \tag{since $(\state,\action') \leaves T$}
		\end{align*}
	
	\item $\state \in T_\circ$:\\
	Since $\action$ is an arbitrary staying action, it suffices to show that for all actions $\action' \in \Av(\state): (\state,\action') \leaves\ T \implies \ub(\state,\action) \leq \ub(\state,\action')$. This is done by the following chain of equations:
		\begin{align*}
		\ub(\state,\action) &= m \tag{by Claim 1}\\
			&\leq \exit<\ub>[\circ](T) \tag{by assumption} \\
			&\leq \min_{\action' \in \Av(\state): (\state,\action') \leaves T} \ub(\state,\action') \tag{since $\state \in T_\circ$}\\
			&\leq \ub(\state,\action') \tag{since $(\state,\action') \leaves T$}
		\end{align*}
\end{enumerate}
\qed
\end{proof}

%% file: convBVIwo.tex
\subsection{Proof of Theorem \ref{lem:convBVIwo}}

\setcounter{theorem}{0}

\begin{theorem}[BVI converges without $\CEC$s]
If the $\SG~\G$ contains no $\CEC$s, then
\[ \forall \state \in \states: \gub[\ast](\state) - \val(\state) = 0,\] 
i.e. value iteration from above converges to the value in the limit.
\end{theorem}

\begin{proof}
We denote the difference in a state by $\diff(\state) \eqdef \gub[\ast](\state) - \val(\state) $.
\begin{enumerate}
\item \label{firstAssm}Assume for contradiction there exists a
    state~$\state \in \states$ s.t. $\diff(\state)>0$.
 \item Let $X \eqdef \set{\state \mid \state \in \states \land
      \diff(\state) =\displaystyle{\max_{\state \in \states}}
      \diff(s)}$ denote the set of all states with maximal
    difference. By definition, we require
    $\lb(\target)=\gub[\ast](\target)=1$ for every $\target \in \F$. Hence,
    $\diff(\target)=0$ and thus, $\F \cap X = \emptyset$. Analogously,
    it follows that the sink state~$\sink \notin X$.
\item \label{claim1} There exists a state $\state<\ell> \in X$ and an
      action~$\action<\ell> \in \Av(\state)$ such that
      $(\state<\ell>,\action<\ell>)$ exits $X$ and it holds:
    \begin{itemize}
    \item (Condition 1) $\gub[\ast](\state<\ell>) =
      \gub[\ast](\state<\ell>,\action<\ell>)$ if $\state \in
      \states<\Box>$
    \item (Condition 2) $\val(\state<\ell>) =
      \val(\state<\ell>,\action<\ell>)$ if $\state \in
      \states<\circ>$
    \end{itemize}
    For proving this statement we distinguish two cases:
  	\begin{enumerate}
	\item \emph{There exists no set of actions~$\actions'$ such
            that $X$ is an end component.} We distinguish two more
          cases.
          \begin{enumerate}
          \item \emph{There exists a partition $X = X_1 \strictunion
              \dots \strictunion X_n$ such that for every $1 \leq i
              \leq n$ there exists a set of action~$\action<i>$ such
              that $X_i$ is an end component.} We apply the proof of
            the second case (starting in Step \ref{case2}) to every
            end component.
          \item \emph{There exists no partition $X = X_1 \strictunion
              \dots \strictunion X_n$ such that for every $1 \leq i
              \leq n$ there exists a set of action~$\action<i>$ such
              that $X_i$ is an end component.} Hence, there exists a
            state~$\state \in X$, which is not part of any end
            component in $X$. Thus, for all actions $\action \in
            \Av(\state)$, it holds $\post(\state,\action)
            \not\subseteq
            X$.\\
            To find a suitable pair $(\state<\ell>,\action<\ell>)$, we
            set $\state<\ell> = \state$ and choose an action
            $\action<\ell>$ such that
            \[ \action<\ell> = \begin{cases}
              \argmax_{\action\in \Av(\state<\ell>)} \gub[\ast](\state<\ell>,\action) &\textnormal{if } \state<\ell> \in X_\Box  \\
              \argmin_{\action\in \Av(\state<\ell>)}
              \val(\state,\action) &\textnormal{if } \state<\ell> \in
              X_\circ,
            \end{cases} \] because all actions in $\state$ exit $X$.
          \end{enumerate}

	\item \label{case2} \emph{There exists a set of
            actions~$\actions'$ such that $X$ is an end component.}
          \begin{enumerate}
          \item $X$ cannot be a $\CEC$, since the game~$\G$ does not
            have any $\CEC$s, which do not contain the target
            state~$\target$ or the sink~$\sink$. Hence, there
            exists a pair $(\state,\action)$ such that $\state \in
            X_\Box$ and $\action \in \Av(\state)$
            s.t. $(\state,\action)$ exits $X$\footnote{As $X$ is not
              a BEC, there must exist both a maximizer and a minimizer
              state with leaving actions. We choose to work with a
              maximizing state here.}. We additionally
            require \[(\state<\ell>,\action<\ell>) =
            \argmax_{\substack{\state' \in X_\Box, \action<\state'>
                \in \Av(\state'): \\(\state',\action<\state'>) \leaves
                X}} \gub[\ast](\state',\action<\state'>).\]
          \item We still need to prove that \[
            \gub[\ast](\state<\ell>,\action<\ell>) =
            \gub[\ast](\state<\ell>) = \max_{\action<m> \in
              \Av(\state<\ell>)} \gub[\ast](\state<\ell>, a_m),\]
            i.e. $\action<\ell>$ maximizes the upper bound of
            $\state<\ell>$.  Assume for contradiction that there is an
            $\action' \in \Av(\state<\ell>)$
            s.t. $\gub[\ast](\state<\ell>,\action') >
            \gub[\ast](\state<\ell>,\action<\ell>)$ and
            $(\state<\ell>,\action')$ stays in $X$. In addition, let
            $m$ be the maximal upper bound occurring in $X$,
            i.e. $\displaystyle m \eqdef \max_{s \in X}
            \gub[\ast](\state)$. \label{subAssmContr} It holds $m >
            \gub[\ast](\state<\ell>,\action<\ell>)$ because
            $\gub[\ast](\state<\ell>) >
            \gub[\ast](\state<\ell>,\action<\ell>)$ and thus,
            $\gub[\ast](\state<\ell>,\action<\ell>)$ cannot be the
            maximal upper bound.  Let $\state \in X$ be a state such
            that $\gub[\ast](\state) = m$. \label{letStateHavem}
        
            \begin{itemize}
            \item If $\state \in X_\circ$, the upper
              bound~$\gub[\ast](\state)$ is defined as the minimal
              upper bound~$\gub[\ast](\state,\action)$ for any
              action~$\action \in \Av(s)$. \label{claim2a} Thus, if
              $\state \in X_\circ$, for all actions $\action \in
              \Av(\state)$, it holds $\gub[\ast](\state,\action) \geq
              m$, since $\gub[\ast](\state)=m$.
          
            \item If $\state \in X_\Box$, there exists an action
              $\action \in \Av(\state)$
              s.t. $\gub[\ast](\state,\action) = m$ and
              $(\state,\action)$ stays in $X$. It holds
              $\gub[\ast](\state) = m$. In addition, we know that $m >
              \gub[\ast](\state',\action')$ for any state~$\state' \in
              X_\Box$ and any action~$\action' \in \Av(\state')$ such
              that $(\state',\action')$ exits $X$ because
              $\gub[\ast](\state<\ell>,\action<\ell>)$ is the largest
              upper bound of a leaving action and still smaller than
              $m$, which is the value of the state. Hence,
              $\gub[\ast](\state)=\gub[\ast](\state,\action)$ for some
              action~\mbox{$\action \in \Av(\state)$}, which stays in
              $X$.
            \end{itemize}
            Let $Y \eqdef \set{\state \mid \state \in X \land
              \gub[\ast](\state)=m}$ be the set of states, which have
            a maximal upper bound in $X$. The previous case
            distinction shows that for every state $\state \in Y$,
            there exists an action~$\action \in \Av(\state)$ such that
            the pair $(\state,\action)$ stays in $X$. We now show that
            there there also exists an action~$\action$ staying in~$Y$.
            \begin{itemize}
            \item To arrive at a contradiction, assume there exists no
              such action, i.e. for all actions~$\action \in
              \Av(\states)$, it holds $(\state,\action)$ exits $Y$.
            \item Hence, $\post(\state,\action) \not \subseteq Y$,
              i.e. there exists a state $\state' \in
              \post(\state,\action) \not \in Y$. It holds
              $\gub[\ast](\state')<m$ by definition of $Y$.
        \item Thus,
          \[ \begin{array}{llr} \gub[\ast](\state,\action) & =
            \displaystyle \sum_{\state[\ast] \in
              \post(\state,\action)}
            \trans(\state,\action)(\state[\ast]) \cdot
            \gub[\ast](\state[\ast]) \\ & = \big(\displaystyle
            \sum_{\state[\ast] \in
              \post(\state,\action)\setminus\set{\state'}}
            \underbrace{\trans(\state,\action)(\state[\ast])}_{\leq 1}
            \cdot \underbrace{\gub[\ast](\state[\ast])}_{\leq
              \gub[\ast](\state)} \big) +
            \underbrace{\trans(\state,\action)(\state')}_{\leq 1}
            \cdot \underbrace{\gub[\ast](\state')}_{<
              \gub[\ast](\state)} & < \gub[\ast](\state),
          \end{array}
          \] which is a contradiction.
        \end{itemize}
      			
        Hence, there exists some set $Z \subseteq Y$ s.t. $Z$ is an
        end component.
      
        By definition of~$Y$ and $Z \subseteq Y$, every state~$\state
        \in Z$ has upper bound~$m$, i.e. $\gub[\ast](\state)=m$.  By
        definition of $X$ and $Z \subseteq Y \subseteq X$, every
        state~$\state \in Z$ has maximal difference. The maximal
        difference~$d$ is defined as $d \coloneqq \diff(\state) =
        \max_{\state[\prime] \in \states} \diff(\state[\prime]) $. As
        elements of the set~$Z$ both have a maximal upper bound and
        the maximal difference, it holds $\val(\state) = m - d$ for
        every $\state \in Z$ (the value is the the upper bound minus
        the difference between the upper bound and the value).

        \begin{itemize}
        \item For all states $\state \in Z_\Box$ and all
          actions~$\action \in \Av(\state)$, it holds
          \[\val(\state,\action) \leq m-d = \val(\state)\] by
          definition of the value and the previous step.
          Hence, \[ \displaystyle \max_{\substack{s \in Z_\Box, a \in \Av(t):\\
              (\state,\action) \leaves Z}} \val(\state,\action) \leq m
          - d. \]
            
        \item For every state $\state \in Z_\circ$ and every action
          $\action \in \Av(\state)$ s.t. $(\state,\action)$ exits
          $Z$, it holds $(\state,\action)$ exits $X$. (For every
          $(\state,\action) \in X_\circ$, it holds
          $\gub[\ast](\state,\action) \geq m$. Either all states
          $\state' \in \post(\state,\action)$ are in $Z$ (and thus,
          have value $m$), which is a contradiction to
          $(\state,\action)$ is leaving $Z$, or there must exist some
          $\state' \in \post(\state,\action)$, which is not in
          $Z$. Assume, it is in $X$. Then, its value must be smaller
          than $m$. Hence, there must be state~$\state[\prime\prime]$,
          which cannot be in $X$, s.t. $\gub[\ast](\state[\prime\prime])
          > \gub[\ast](\state)=m$. Hence, if an state-action-pair is
          leaving $Z$, it is also leaving $X$).
        \end{itemize}
        
        For every state $\state \in Z_\circ$ and every action $\action
        \in \Av(\state)$ s.t. $(\state,\action)$ exits $Z$, it holds
        $\val(\state,\action) \geq m - d$.  Additionally,
        $(\state,\action)$ exits $X$, and hence, \[
        \gub[\ast](\state,\action) - \val(\state,\action) < d. \]
        Since $\displaystyle d > \gub[\ast](\state,\action) -
        \val(\state,\action) \geq m - \val(\state,\action)$, it holds
        that \[\val(\state,\action) > m - d.\] This implies that for
        every state~$\state \in Z_\circ$, and every action~$\action
        \in \Av(\state)$ s.t. $(\state,\action)$ exits $Z$, it
        holds \[
        \val(\state,\action) > \max_{\substack{\state[\ast] \in Z_\Box, \action[\ast] \in \Av(t):\\
            (\state[\ast],\action[\ast]) \leaves Z}}
        \val(\state[\ast],\action[\ast]). \] Thus, every exiting
        action of the minimizing player has a higher value than the
        best exiting action of the maximizing player, which means that
        the end component is bloated. Hence, $Z$ is a
        BEC, which is a contradiction to the assumption that the game
        does not contain any bloated end component. Thus, the
        assumption was wrong, and it holds that $\action<\ell>$
        maximizes the upper bound of $\state<\ell>$.
      \end{enumerate}
    \end{enumerate}

  \item \label{finallyDeriveContr}We now use Step~\ref{claim1} to arrive at a contradiction to
    the assumption in Step~\ref{firstAssm} as follows: If
    $\state<\ell> \in X_\Box$, it holds $\gub[\ast](\state<\ell>) =
    \gub[\ast](\state<\ell>,\action<\ell>)$ by Step~\ref{claim1} and
    $\val(\state<\ell>) \geq \val(\state<\ell>,\action<\ell>)$ by
    Equation~\ref{eq:Vs} ($\state<\ell>$ is a \emph{maximizing}
    state).\label{deriveBox} If $\state<\ell> \in X_\circ$, it holds
    $\val(\state<\ell>) = \val(\state<\ell>,\action<\ell>)$ by
    Step~\ref{claim1} and $\gub[\ast](\state<\ell>) \leq
    \gub[\ast](\state<\ell>,\action<\ell>)$ ($\state<\ell>$ is a
    \emph{minimizing} state).\label{deriveCirc} Hence, we know
  \begin{align*}
    \diff(\state<\ell>)
    &= \gub[\ast](\state<\ell>) - \val(\state<\ell>) \tag{by definition of $\diff$}\\
    &\leq \underbrace{\gub[\ast](\state<\ell>,\action<\ell>)}_{\geq
      \gub[\ast](\state<\ell>)} -
    \underbrace{\val(\state<\ell>,\action<\ell>)}_{\leq
      \val(\state<\ell>)} \tag{by Step \ref{deriveBox}} \\
    &= \diff(\state<\ell>,\action<\ell>) \tag{by definition of $\diff$}\\
    &< \diff(\state<\ell>) \tag{$\action<\ell>$ exits $X$}
    \end{align*} 
    We arrive at a contradiction. Thus, our assumption was wrong and
    it holds for all states $\state \in \states: \diff(\state)=0$.\qed
  \end{enumerate}
\end{proof}

%% file: find.tex
\subsection{Proof of Lemma \ref{lem:find}}

\setcounter{lemma}{1}

\begin{lemma}[Correctness of Algorithm \ref{alg:FIND}]
	  $T \in$ \emph{\FIND}$(\val)$ if and only if $T$ is an MSEC.
\end{lemma}
\begin{proof}
We prove the equivalence by showing each direction separately.
\begin{itemize}
\item  $T \in$ \FIND$(\val)$ $\implies$ $T$ is an MSEC
\begin{enumerate}
\item Let $T \in$ \FIND$(\val)$.
So $T$ is a MEC of the game where the mapping of available actions $X \eqdef \{(\state,\{\action\in\Av(\state) \mid \val(\state,\action) > \val(\state)\}) \mid \state \in \states<\circ>\}$ has been removed. 
\item Hence for all $\state \in T_\circ, \action \in \Av'(\state): \val(\state,\action) = \val(\state)$. \label{step:allCircEqual}

\item We now show that for all $\state \in T: \val(\state) = \exit<\val>[\Box](T)$.
\begin{itemize}
\item If there is no exit for player $\Box$, then $\exit<\val>[\Box](T) = 0$, and all states in $T$ have the value 0, since player $\circ$ can force the game to stay inside the EC forever.
\item If there is an \leaving\ state $e$ for player $\Box$, then also from each state in $T$ there is a path using only states in $T$ to $e$.
For each state $\state \in T$ we can compute $\val(\state)$ by recursively applying the Bellman equations. In each application, we choose an action that leads us closer to $e$, i.e. for $t$ being the next state on the path to $e$, choose $\action$ such that $\val(\state,\action) = \trans(\state,\action,t) \cdot \val(t) + \sum_{\state' \in \post(\state,\action) \setminus \{t\}} \trans(\state,\action,\state') \cdot \val(\state') $. Since for each state in $T$ we have a path to $e$, we can replace each $\val(\state')$ in this way. By repeating this and thereby multiplying all the probabilities $\trans(\state,\action,t)$, the factor in front of the terms that do not contain $\val(e)$ approach 0, and thus we get $\val(\state,\action)=\val(e)$ for an arbitrary $\state \in T$ and a staying action $\action$. Since for the minimizer all actions have the same value by Step \ref{step:allCircEqual}, all minimizer states have as value $\val(e)$. The maximizer can certainly achieve $\val(e)$, since he picks the action with the maximal value. He can not achieve more than $\val(e)$, because this is defined to be the best exit from the EC, and thus the best value that can be achieved from any state in the EC.
\end{itemize}

\item Since $T$ is a MEC of the game with some available actions removed, it certainly is an EC in the original game. Thus, from this and the previous step, we know that $T$ is a SEC.

\item $T$ is an MSEC, because if there was some $T' \varsupsetneq T$, such that $T'$ is a SEC, then there exists some state $\state \in T' \setminus T$ with all staying actions of this state having the value $\exit<\val>[\Box](T)$. If it is a minimizer state, it cannot have a lower exit available, because otherwise $T'$ would not be a SEC. Thus no staying action of $\state$ is removed by  Line \ref{line:findRemove}, and it should also be part of the MEC in the modified game. This contradicts the assumption that $\state \notin T$, and thus $T$ is an MSEC.
\end{enumerate}

\item $T$ is an MSEC $\implies$ $T \in$ \FIND$(\val)$
\begin{enumerate}
\item Let $T$ be an MSEC. 
We need to show that $T$ is a MEC of the game where the set of actions $X \eqdef \{(\state,\{\action\in\Av(\state) \mid \val(\state,\action) > \val(\state)\}) \mid \state \in \states<\circ>\}$ has been removed.

\item Since $T$ was an inclusion maximal component in the original game and by removing actions we cannot make an end component larger, if $T$ is an EC in the modified game, it is a MEC.

\item So we now show that $T$ is a MEC in the modified game:\\

Since $T$ is a MEC in the original game, there exists a $B$ such that for each $\state \in T, \action \in B \cup \Av(\state)$ we do not have $(\state,\action) \leaves T$.
Hence for all $\state \in T, \action \in B \cup \Av(\state)$ it holds that
\begin{align*}
\val(\state,\action) &=  \sum_{s' \in S} \trans(\state,\action,\state') \cdot \val(\state') \tag{by Equation \ref{eq:Vsa}} \\
&= \sum_{s' \in S} \trans(\state,\action,\state') \cdot \exit<\val>[\Box](T) \tag{since $\state' \in T$ and $T$ is simple}\\
&= \exit<\val>[\Box](T) 
\end{align*}
So every action that stays in $T$ (i.e. does not exit $T$), in particular every action in $B$, is not removed from the game, because for all $\state \in T, \action \in \Av(\state), \neg (\state,\action) \leaves T: \val(\state) = \exit<\val>[\Box](T) = \val(\state,\action)$. The first equality comes from $T$ being an MSEC, the second is what we have just shown.

Since no action in $B$ is removed from the game, $T$ is still a MEC after the removal, and hence $T \in$ \FIND$(\val)$ .
\end{enumerate}
\end{itemize}
\qed
\end{proof}

%% file: corrADJ.tex
\subsection{Proof of Lemma \ref{lem:corrADJ}}\label{app:corrADJ}

\setcounter{lemma}{2}

\begin{lemma}[$\ADJUST$ is sound]
For any $f:\states\to[0,1]$ such that $f\geq\val$ and any EC $T$,
$\ADJUST(T,f)\geq\val$.
\end{lemma}

\begin{proof}
Let $T$ be an EC and $f:\states\to[0,1]$ such that $f\geq\val$. 
\begin{enumerate}
\item We reformulate the goal to saying that for all state $\state \in T$ it holds that $\min(f(\state), \exit<f>[\Box](T)) \geq \val(\state)$. \\
This is equivalent to our goal, because the change in Line \ref{line:adjust} is the only change the Algorithm \ref{alg:adjust} applies and because the comparison of the functions $\ADJUST(T,f)$ and $\val$ is pointwise.
\item If in Line \ref{line:adjust} of $\ADJUST$ $\min(f(\state), \exit<f>[\Box](T))$ gets evaluated to $f(\state)$ or if $f(\state)=\exit<f>[\Box](T)$, by assumption of $f\geq\val$ the goal trivially holds.
\item If $f(\state) > \exit<f>[\Box](T)$,  the following chain of equations proves our goal: 
	\begin{align*}
	\ADJUST(T,f)(\state)& =\exit<f>[\Box](T)\\
	& \geq \exit<\val>[\Box](T) \tag{Since $f\geq\val$}\\
	& \geq \val(\state) \tag{Since no state can achieve a greater value than the best exit}
	\end{align*}
	
\end{enumerate}
\qed
\end{proof}

%% file: BVIA_term.tex
\subsection{Proof of Theorem \ref{BVIA_term}}

\setcounter{theorem}{1}

\begin{theorem}[Soundness and completeness]\label{BVIA_term}
Algorithm~\ref{alg:gen} (calling Algorithm~\ref{alg:BVIA}) produces monotonic sequences $\lb$ under- and $\ub$ over-approximating $\val$, and terminates.
\end{theorem}
\begin{proof}
We denote by $\lb<i>$ and $\ub<i>$ the lower/upper bound function after the $i$-th call of $\UPDATE$.
\begin{enumerate}
\item The fact that $\lb<i>$ and $\ub<i>$ are monotonic under- respectively over-approximations of $\val$ comes from the fact that they are updated via Bellman updates, which preserve monotonicity and the under-/over-approximating property (this can be shown by a simple induction), and from Lemma \ref{lem:corrADJ}.
\item We still need to prove that the algorithm terminates, i.e. that for all $\epsilon > 0$ there exists an $n$ such that $\ub<n>(\initstate)-\lb<n>(\initstate) < \epsilon$ .\\
\item Since $\lim_{n \to \infty} \lb<n> = \val$ (from e.g. \cite{visurvey}), it suffices to show that $\lim_{n \to \infty} \ub<n> - \val = 0$. In the following, let $\diff(\state) \eqdef \lim_{n \to \infty} \ub<n> - \val$.
\item Assume for contradiction that there exist states with $\diff(\state)>0$.\label{assmContr}
\item Let $\displaystyle X \eqdef  \set{s \mid \state \in \states \land \diff(\state) = \max_{s \in \states} \diff(\state)}$ be the set of those states that have the maximal difference. 
\item Since $\diff > 0$ is required by Step \ref{assmContr}, $\set{\target,\sink} \cap X = \emptyset$. 
\item Hence $X \subset \states$ and there have to be states outside of $X$. 
\item If $X$ is not a SEC, the proof of Theorem \ref{lem:convBVIwo} proves our goal. So we continue with the assumption that $X$ is a SEC.
$X$ also is an MSEC, because the upper bound is propagated to all states in $X$, as all of them can achieve the bound of the best exit.
\item At some point $\lb<i>$ converges close enough to $\val$ to fix all decisions of the Minimizer correctly, and then $X \in \FIND(\lb<i>)$.
	\begin{itemize}
	\item This need not happen surely, but we can have the case where two actions have equal value, but different $\lb$ forever. But that case is no problem, since then the part that seems to have the smaller value is deflated, and the other sets its upper bounds accordingly.
	\end{itemize}
\item If there was no leaving state-action pair for player $\Box$, $\ub(\state)=0$ for all states in $X$ by Line \ref{line:adjust} of $\ADJUST$, and hence the difference would be 0. Thus we continue with the assumption that there exists a leaving state-action pair for player $\Box$.
\item Without loss of generality, let $\state<\ell> \in X_\Box, \action<\ell> \in \Av(\state<\ell>)$ be a leaving state-action pair with
\[\lim_{n \to \infty} \ub<n>(\state<\ell>,\action<\ell>) = \exit<\ub>[\Box](X).\]
\item If $\ub(\state<\ell>,\action<\ell>) = \ub(\state<\ell>)$, we are done by the same chain of equations as in Step \ref{finallyDeriveContr} of the proof of Theorem \ref{lem:convBVIwo}.
\item If $\ub(\state<\ell>,\action<\ell>) \neq \ub(\state<\ell>)$, it holds that $\ub(\state<\ell>,\action<\ell>) < \ub(\state<\ell>)$ by Equation \ref{eq:Ls} and since $\state<\ell> \in X_\Box$. 
\item Then, since $\ub(\state<\ell>,\action<\ell>)$ is the maximal leaving upper bound, all upper bounds in $X$ get decreased to it by \ref{line:adjust} of $\ADJUST$.  This is a contradiction, because $\diff$ is defined as the limit of the difference. Thus, our assumption was wrong and it holds for all states
    $\state \in \states: \diff(\state)=0$.
\end{enumerate}
Note that all our algorithms take finite time, as they are monotonic on a finite domain, i.e. we always execute something for each member of a finite set.
\qed
\end{proof}